\documentclass[reprint,amsmath,amssymb, aps, pra]{revtex4-2}
\usepackage[english]{babel}
\usepackage[utf8]{inputenc}
\usepackage[T1]{fontenc}

\usepackage{amsthm}
\usepackage{amsmath,amssymb}
\theoremstyle{plain} %% for non-numbering
\usepackage{graphicx}
\usepackage[caption=false]{subfig} %% for subfigures
\usepackage{float} %% for positioning of the figures
\usepackage[colorlinks=true, allcolors=blue]{hyperref}
\usepackage{braket}
\usepackage{float}
\usepackage{bbm,amsfonts}
\usepackage[shortlabels]{enumitem}
\usepackage[dvipsnames]{xcolor}
\usepackage[title]{appendix}
\usepackage{dsfont} %% double stroke for sets
\usepackage{tikz}
\usepackage{rotating} %% for rotated text
\usepackage{mathtools} %% for good define sign ":="

\hypersetup{
	colorlinks=true,  
	linkcolor=blue,   
	citecolor=blue,   
	urlcolor=blue     
}

\newtheorem{theorem}{Theorem}

\newtheorem*{example*}{Example}

\newtheorem{definition}{Definition}

\newtheorem*{remark*}{Remark}
\newtheorem{corollary}[theorem]{Corollary}

\usepackage{soul}
\usepackage{xcolor}

\newcommand{\Od}[1]{{\mathcal{O}(#1)} }
\usetikzlibrary{quantikz2}
\usepackage{circuitikz}
\usepackage{multirow}
\DeclareUnicodeCharacter{2061}{}
\usepackage{tikz}

\usepackage{centernot}
\usepackage{mathtools}
\usepackage{stmaryrd}
\makeatletter
\newcommand{\xMapsto}[2][]{\ext@arrow 0599{\Mapstofill@}{#1}{#2}}
\def\Mapstofill@{\arrowfill@{\Mapstochar\Relbar}\Relbar\Rightarrow}
\makeatother

\begin{document}
\title{Quantum algorithm for solving generalized eigenvalue problems with application to the Schr\"odinger equation}
\author{Grzegorz Rajchel-Mieldzio\'c}

\author{Szymon Pli\'s}
\author{Emil Zak}
\affiliation{BEIT sp.\ z o.o., ul.\ Mogilska 43, 31-545 Krak{\'o}w, Poland}

\date{\today}

\begin{abstract}
Accurate computation of multiple eigenvalues of quantum Hamiltonians is essential in quantum chemistry, materials science, and molecular spectroscopy. Estimating excited-state energies is challenging for classical algorithms due to exponential scaling with system size, posing an even harder problem than ground-state calculations.
We present a quantum algorithm for estimating eigenvalues and singular values of parameterized matrix families, including solving generalized eigenvalue problems that frequently arise in quantum simulations. Our method uses quantum amplitude amplification and phase estimation to identify matrix eigenvalues by locating minima in the singular value spectrum.
We demonstrate our algorithm by proposing a quantum-computing formulation of the pseudospectral collocation method for the Schrödinger equation.
We estimate fault-tolerant quantum resource requirements for the quantum collocation method, showing favorable scaling in the size of the problem $N$ (up to $\widetilde{\mathcal{O}}(\sqrt{N})$) compared to classical implementations with $\widetilde{\mathcal{O}}(N)$, for certain well-behaved potentials. Additionally, unlike the standard collocation method, which results in a generalized eigenvalue problem requiring matrix inversion, our algorithm circumvents the associated numerical instability by scanning a parameterized matrix family and detecting eigenvalues through singular value minimization. This approach is particularly effective when multiple eigenvalues are needed or when the generalized eigenvalue problem involves a high condition number. In the fault-tolerant era, our method may thus be useful for simulating high-dimensional molecular systems with dense spectra involving highly excited states, such as those encountered in molecular photodynamics or quasi-continuum regimes in many-body and solid-state systems.
\end{abstract}

\maketitle

\section{Introduction}
Quantum computing algorithms are sequences of multi-qubit state transformations in quantum devices. Certain quantum algorithms, such as Quantum Phase Estimation (QPE), Shor's algorithm, and Grover's algorithm~\cite{Nielsen_Chuang}, have a proven improved computational complexity scaling for specific problems compared to their classical counterparts~\cite{Lee2023}. In the long term, this theoretical advantage may enable large, fault-tolerant quantum devices to outperform classical computers, provided that current technological challenges in building quantum hardware are overcome. However, the utility of these devices will ultimately depend on the ability of quantum algorithms to address scientifically and industrially relevant problems.

Notably, simulating quantum-mechanical systems remains challenging for classical computers due to the curse of dimensionality, where both memory requirements and computational time scale exponentially with the number of particles or degrees of freedom. In such simulations, calculating the eigenvalues (spectra) of the Hamiltonian is often required. These eigenvalues and their corresponding eigenvectors are fundamental for studying natural phenomena and predicting material properties. For this reason, quantum Hamiltonian simulation is a promising area for demonstrating quantum computational advantage over classical architectures.

To date, much of the research has focused on ground state energy calculations~\cite{su2021,babbush2018,lee2021,burg2021,Lee2023,Deka2025}. However, this focus limits the scope of potential quantum advantage. Excited state calculations, i.e., computing eigenvalues beyond the ground state, become rapidly too complex for classical algorithms, suggesting that quantum advantage may be more pronounced in such scenarios.
As an example, determining multiple eigenvalues of a quantum-mechanical Hamiltonian is essential for modeling physical phenomena such as intraband vibrational relaxation~\cite{Richerme2023,Mariano2025}, absorption spectra~\cite{Bunker1999}, thermodynamic properties of crystals~\cite{Kittel1965}, and (photo)chemical reaction pathways~\cite{Weight2023,Mukherjee2022}, among others. In these cases, standard QPE becomes inefficient. The number of QPE calls required to resolve multiple excited state energies scales poorly, and the circuit complexity for preparing suitable inputs tailored to excited states increases rapidly. For this reason, beyond ground-state calculations, further theoretical work is needed to clarify potential quantum advantage.

In this work, we propose a quantum algorithm for computing multiple eigenvalues of general operators. Our central idea is to encode a grid of matrix parameters $\alpha$ into a quantum superposition, thereby defining a family of single-parameter-dependent matrix representations $\widetilde{\mathbf{M}}(\alpha)$ of a general operator. By combining QPE with quantum amplitude amplification (AA), we selectively enhance the probability of measuring quantum states corresponding to eigenvalues of $\widetilde{\mathbf{M}}(\alpha)$ within a prescribed interval, from which the associated parameter values $\alpha$ can be extracted.

In the context of the Schrödinger equation, the matrix family takes the form $\mathbf{M}(\alpha)=\mathbf{H}-\alpha\mathbf{S}$, where $\mathbf{H}$ is the Hamiltonian, $\mathbf{S}$ is the Gram (overlap) matrix associated with a chosen basis, and $\alpha$ represents a candidate eigenvalue. By qubit-encoding $\alpha$ and amplifying the amplitudes of quantum states corresponding to values of $\alpha$ for which $\widetilde{\mathbf{M}}(\alpha)$ has singular values within $\varepsilon$ from 0, we obtain a quadratic reduction in Clifford+T gate complexity compared with a direct application of QPE.

Unlike conventional ground-state Hamiltonian simulation approaches based on QPE \cite{babbush2018}, which rely critically on high overlap between an initial trial state and the target eigenstate \cite{Shao_2022,Kerzner_2024,Low_2024,Ding2024}, our method scans a family of matrices to detect eigenvalues or singular values without requiring significant initial overlap. A related strategy was proposed by Kerzner et al. \cite{Kerzner_2024}, who used amplitude amplification and estimation to identify matrix eigenvalues below a specified threshold. By contrast, our approach is designed to determine parameter values for which a matrix in a single-parameter family possesses an eigenvalue within a chosen interval.

The key conceptual distinction from earlier eigenvalue-estimation techniques \cite{Jin_2020,Chen_2024,Kerzner_2024} is that the parameter of interest, such as the energy, is encoded directly in the quantum state rather than in the measured eigenvalue register. Consequently, the probability of observing a given bitstring corresponding to an approximate energy is proportional to the squared, amplified amplitude. Within this framework, interesting eigenvalues $\alpha$ correspond to minima in the singular-value landscape of the matrices $\widetilde{\mathbf{M}}(\alpha)$, and the algorithm enhances the measurement probability associated with these minima.

An important practical feature of our algorithm is that it does not incur additional overhead from repeated sampling to estimate eigenvalues: each execution yields a single eigenvalue estimate with precision $\varepsilon$. This contrasts with variational quantum eigensolver (VQE) approaches, where the number of measurement shots scales unfavorably and often limits practical performance and demonstrations of quantum advantage \cite{Zimboras,Tilly2022}.

As a concrete application, we apply the algorithm to the Schrödinger equation. This leads to our second contribution: the formulation of a quantum-computational analogue of the collocation method for solving partial differential equations \cite{Peet1989,Helgaker2000,Carrington2017,Manzhos2023}. Collocation is a pseudospectral (grid-based) approach that reformulates the differential Schrödinger equation as a generalized, generally rectangular and non-symmetric, matrix eigenvalue problem. Such generalized eigenvalue problems are popular in quantum physics and chemistry and therefore have broad relevance.

Generalized eigenvalue problems can be written as
\begin{equation}
    \mathbf{H}\mathbf{v} = \mathbf{S}\mathbf{v}\mathbf{E},
    \label{eq:gen_eval}
\end{equation}
where $\mathbf{S}$ is the Gram matrix \cite{Helgaker2000} and $\mathbf{v}$ is the matrix of eigenvectors. These problems are typically more challenging than standard eigenvalue problems \cite{Hu1998,Brown2015,Hashemi2022}, in part because they require inversion or factorization of the Gram matrix. Inverting an $N\times N$ matrix requires $\mathcal{O}(N^3)$ floating-point operations, while the condition number of $\mathbf{S}$ controls numerical stability and may necessitate high-precision arithmetic \cite{recipes}. Such problems occur frequently in quantum chemistry \cite{Ford1974,Helgaker2000,Dusson} and in vibrational and vibronic structure calculations \cite{Carrington2017,Szalay2011,gaborig,Poirier2000,Brown2015,Richings2015,Burghardt2008,Hu1998}, where computational convenience often comes at the expense of a non-identity overlap matrix.
Collocation methods based on Gaussian basis functions provide a flexible and efficient framework for solving the Schrödinger equation by avoiding high-accuracy quadrature and explicit integral evaluation, and by allowing unconstrained grid selection \cite{Manzhos2023}. However, they often produce ill-conditioned Gram matrices due to near-linear dependence of basis functions, leading to large condition numbers and limiting scalability. As a result, despite progress \cite{Carrington2017,Zak2019,Carrington2021,Simmons2023}, the classical computational cost of solving the resulting eigenvalue problems still suffers from the curse of dimensionality.

Our approach circumvents explicit matrix inversion by instead studying a family of parameterized matrices and using the smallest singular value as a diagnostic of eigenvalue accuracy. Scanning a grid of candidate parameters provides an alternative route to directly solving generalized eigenvalue problems, applicable not only to the Schrödinger equation but also to broader classes of single-parameter matrix families, including nonlinear parameterizations. In Sec.~\ref{sec:application}, we compare our method with classical collocation approaches and highlight its advantages.

In Sec.~\ref{subsec:algorithms_complexities}, we estimate the quantum resources required to implement our algorithm on a fault-tolerant quantum computer, including logical qubit counts and T-gate complexity. These requirements are compared with classical resources needed for SVD-based energy-grid scanning methods, allowing us to identify regimes in which quantum advantage may emerge. By applying our algorithm to the collocation formulation of the Schrödinger equation, we demonstrate favorable scaling in both memory usage and gate complexity relative to the best-known classical approaches \cite{Carrington2017}, owing to quantum parallelism introduced by amplitude amplification on a superposed quantum state.

The remainder of the paper is organized as follows. In Sec.~\ref{sec:general_method}, we introduce the general quantum algorithm for eigenvalue estimation via parameterized matrix families. In Sec.~\ref{sec:application}, we describe application of our method to the collocation formulation of the Schrödinger equation. We then compare classical and quantum resources and performance in different regimes, followed by a discussion in Sec.~\ref{sec:discussion} and concluding remarks in Sec.~\ref{sec:conclusions}.

\section{Quantum algorithm for landscape scanning}\label{sec:general_method}
In order to devise an algorithm that solves the generalized eigenvalue problem presented in Eq.~\eqref{eq:gen_eval}, we first consider a broader problem: finding the \textit{pseudospectrum} of an operator, defined as follows:

\begin{definition}[\cite{Hogben_2013}]
The pseudospectrum $\Lambda_\varepsilon$ (more precisely, the $\varepsilon$-pseudospectrum) of an operator $\mathbf M$ is defined as the set of all numbers $\lambda$, for which there exists an operator $\mathbf X$ such that $||\mathbf X - \mathbf M|| \leq \varepsilon$ and $\lambda$ is an eigenvalue of $\mathbf X$.
\end{definition}

In the pseudoscpectrum finding context we consider an $N$-dimensional one-parameter family of 
matrices ${\mathbf M(\alpha)}$, with $\alpha \in A$, where $A$ is a connected subset of the real line, $A = [a, b] \subset \mathbb{R}$, where we denote the range $S=b-a$. Matrices can be general. In this work, we consider two limiting cases of square Hermitian and rectangular~\footnote{In this case, $N$ is the larger of the two dimensions.} matrices.
Our goal is to identify at least one matrix $\mathbf M(\alpha_0)$ that possesses a specific eigenvalue~\footnote{Later on, we generalize this problem also to singular values.} $\lambda_0$ determined with precision $\epsilon$; that is, $\lambda_0$ belongs to the pseudospectrum of $\mathbf M(\alpha_0)$, $\lambda_0 \in \Lambda_\varepsilon(\mathbf M(\alpha_0))$~\footnote{For normal matrices, $\varepsilon$-pseudospectrum are the regions $\pm\varepsilon$ around their spectrum; therefore, this notion is aligned with our goal of finding numbers close to the exact eigenvalues.}.
We assume that such an $\alpha_0$ exists.

\textbf{Square Hermitian matrices.}
First, we consider a family of square Hermitian matrices $\mathbf M(\alpha)$ that admit a truncated series expansion of the form
\begin{equation}
\mathbf M(\alpha) = \sum_{j=0}^J \alpha^j \mathbf M_j,
\label{eq:M-family}
\end{equation}
In the particular case of $J = 1$, the generalized eigenvalue problem written in Eq.~\eqref{eq:gen_eval} can be cast as a pseudospectrum-finding problem, by taking $\mathbf M(\alpha) = \mathbf{H} - \alpha \mathbf{S}$ and searching for an $\alpha$ that minimizes magnitude of the smallest eigenvalue $|\lambda_0|$ of $\mathbf{M}(\alpha)$. We discuss this special case in Sec.~\ref{sec:quantum-landscape}.

For implementation on a quantum computer, we discretize the parameter $\alpha$ over an equidistant grid ${\alpha}_{i \in {1, \ldots, K}} = \{a, a+\delta, \ldots, b-\delta, b\}$.
The family of matrices $\mathbf M(\alpha)$ given in Eq.~\eqref{eq:M-family} can be  then conveniently extended as follows:
\begin{equation}
\label{eq:firstMtylda}
\begin{split}
\widetilde{\mathbf M} &= \sum_{j=0}^J \mathbf D^j(\alpha) \otimes \mathbf M_j = \\
&= \mathbb{I} \otimes \mathbf M_0 + \mathbf D(\alpha) \otimes \mathbf M_1 + \mathbf D^2(\alpha) \otimes \mathbf M_2 + \ldots,
\end{split}
\end{equation}
where $\mathbf D(\alpha) \coloneqq \text{diag}(\alpha_1, \ldots, \alpha_K)$ is a diagonal matrix transforming states inside a $k\coloneqq \log K$-qubit Hilbert space $\mathcal{H}_k$, i.e., $\mathbf D^j(\alpha): \mathcal{H}_k\rightarrow \mathcal{H}_k$, where $\mathcal{H}_k = \mathbb{C}^{2^k}$. Similarly, $M_j: \mathcal{H}_n\rightarrow \mathcal{H}_n$, where $\mathcal{H}_n = \mathbb{C}^{2^n}$, such that $\widetilde{\mathbf M}: \mathcal{H}\rightarrow \mathcal{H}$, i.e., $\mathcal{H}=\mathcal{H}_k\otimes\mathcal{H}_n$, and $\dim \mathcal{H}=N\cdot K$. 
We aim to find the eigenvalues of $\mathbf M(\alpha)$ in a specific region $I = [\lambda_0 - \varepsilon,\lambda_0 + \varepsilon]$ and the corresponding values for $\alpha$, as schematically depicted in Fig.~\ref{fig:landscape_sketch}. 

   \begin{figure}[h!]
       \centering
       \includegraphics[width=1\linewidth]{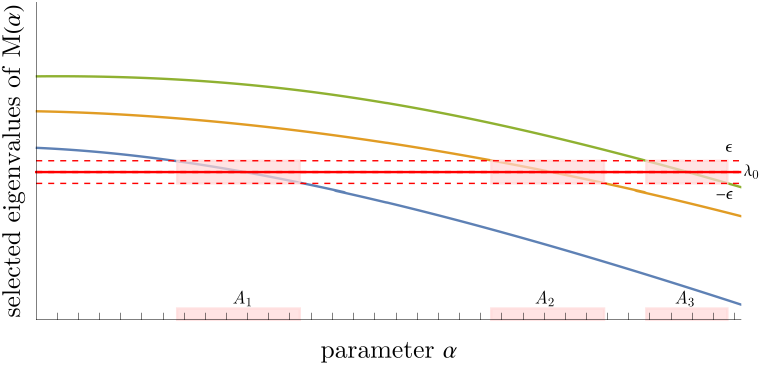}
       \caption{A sketch of the generic landscape of the eigenvalues for a one-parameter family of Hermitian matrices $\mathbf M(\alpha)$. Here, for all equidistant grid values of $\alpha \in A = A_1 \cup A_2 \cup A_3$, the matrices $\mathbf M(\alpha)$ have eigenvalue $\lambda_0$ with accuracy $\varepsilon$.}
       \label{fig:landscape_sketch}
   \end{figure}

Our proposed algorithm is summarized in the following theorem:
    \begin{theorem}\label{thm:quantum_algorithm_square}
        Consider a discretized one-parameter family of Hermitian matrices $\mathbf M(\alpha)$ on a grid formed by $K$ points. 
        If the dependence on $\alpha$ can be written as a finite series expansion $\mathbf M(\alpha) = \sum_{j=0}^J \alpha^j \mathbf M_j$, for some matrices $\mathbf M_j$, there exists a quantum algorithm that finds a set of the parameters $\alpha'$ such that matrices $\mathbf M(\alpha')$ have eigenvalues in a region $[\lambda_0-\varepsilon,\lambda_0+\varepsilon]$ with $\widetilde{\mathcal{O}}(J\zeta\sqrt{NK}/\varepsilon)$ calls to block-encoded matrices $\mathbf M_j$, which all have their max-norms upper-bounded by $\zeta$.
    \end{theorem} 

\begin{proof}

Quantum circuit representing our key elements of the algorithm is shown in Figure~\ref{fig:circuit-A}.
Our procedure begins by constructing the following state, similar to that in Ref.~\cite{Kerzner_2024}:
\begin{equation}\label{eq:init_state}
\begin{split}
    \ket{\Psi}&=\mathrm{PREP}\left(\ket{0}_{k}\otimes \ket{0}_n\right)=\frac{1}{\sqrt{NK}}\sum_{i=0}^{2^{(n+k)}-1}\ket{\psi_i}\ket{\bar\psi_i} \\
    &= \frac{1}{\sqrt{NK}}\sum_{i=0}^{2^{(n+k)}-1}\ket{ii} \in\mathcal{H}\otimes\mathcal{H}
\end{split}
\end{equation}
where $\ket{\psi_i} = \ket{\alpha_i}_k\otimes\ket{\phi_i}_n$ is the $i$-th eigenstate of $\widetilde{\mathbf M}$, as defined in Eq.~\eqref{eq:firstMtylda}, and $\ket{ii}$ is composite state keeping values of integer indices $i=1,2,...,N K$.
The operator $\mathrm{PREP}$ acts as a preparation phase of a maximally entangled state between these two registers.
The latter equality in Eq.~\eqref{eq:init_state} holds for any choice of basis $\ket{\psi_i}$.
In our notation, eigenvector $\ket{\psi}$ corresponding to a priori chosen eigenvalue $\lambda_0$ of $\widetilde{\mathbf M}$ is denoted as $\ket{\psi}=\ket{\psi_0}$.
The statevector given in Eq.~\eqref{eq:init_state} can be straightforwardly constructed with a circuit consisting of Clifford CNOT and Hadamard gates.

\onecolumngrid
\begin{center}
\begin{figure}[H]
       \centering
    \includegraphics[width=\linewidth]{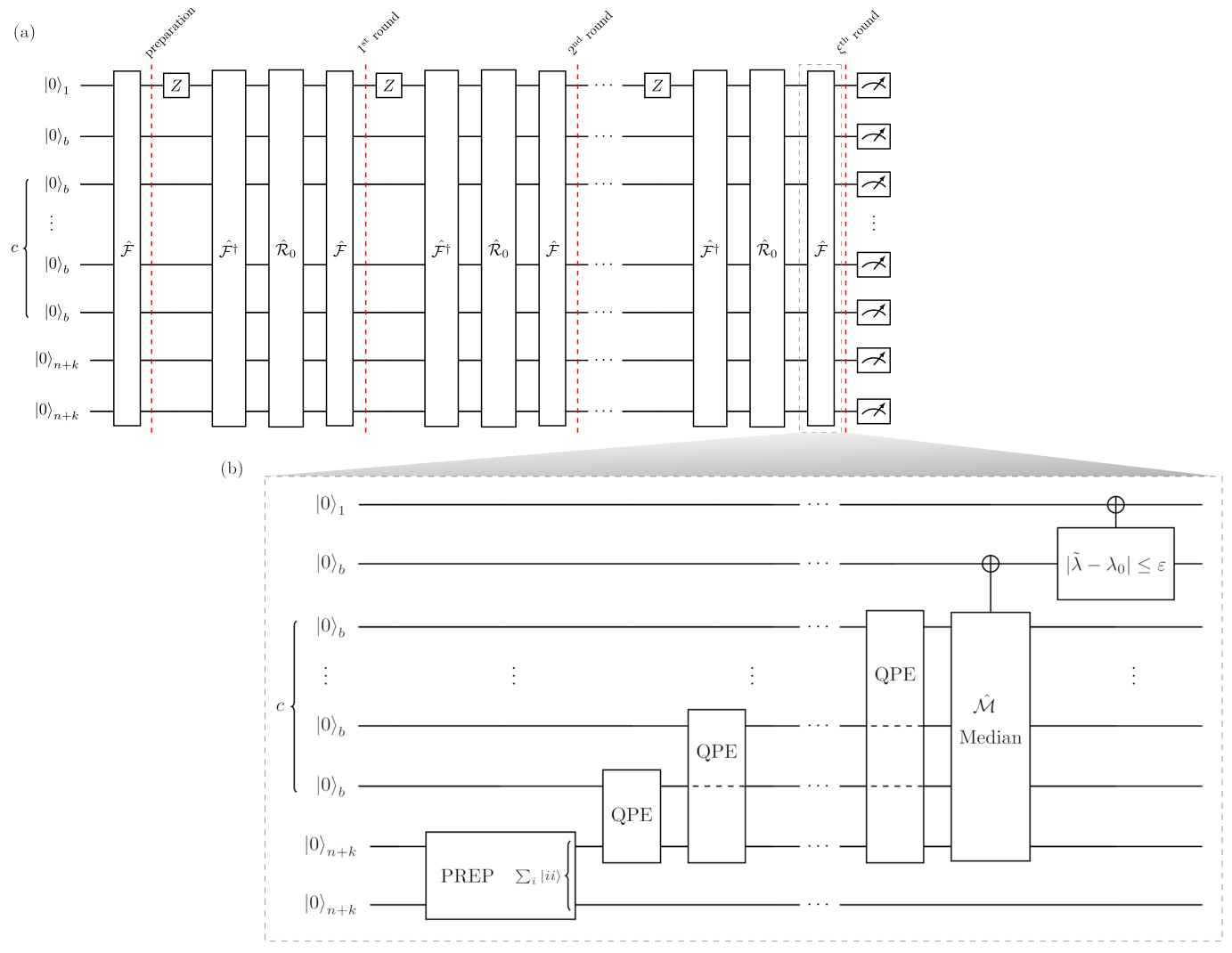}
    \caption{A quantum circuit realizing the search for eigenvalues of a matrix that are close to the target value $\lambda_0$, effectively constructing the algorithm described in Theorem~\ref{thm:quantum_algorithm_square} for $N = 2^n$ and $K = 2^k$.
    The costs of the block-encoding are assumed to depend on the technique, and for example, for $d$-sparse they will scale as $\mathcal{O}(\log NK+1+\beta)$, where $\beta$ is the precision for the elements of the matrix $\mathbf M$ in the chosen block-encoding.
    The entire circuit (a) is a modification of the standard Grover algorithm. 
    First, via $c$  steps of Quantum Phase Estimation, we prepare a state in an equal superposition of the eigenvectors and the corresponding eigenvalues on $c+1$ registers of $b$ qubits each. 
    The ``median'' register stores the best approximation to the true eigenvalues~\cite{Kerzner_2024}.
    Additionally, we apply the region oracle on the uppermost qubit, which gives $\ket{1}$ if the corresponding eigenvalue is in the region of interest and $\ket{0}$ otherwise.
    This part of the procedure is called $\hat{\mathcal{F}}$, depicted on part (b).
    Observe that its action is equivalent to the tensor product of Hadamard gates in the standard Grover algorithm.
    Having prepared the state, in the first round, the oracle $Z$ acting on the qubit with the characteristic function $\chi$ is applied. 
    What follows is the reflection over the equal superposition state, realized via $\hat{\mathcal{F}}^\dagger\hat{\mathcal{R}}_0\hat{\mathcal{F}}$, where $\hat{\mathcal{R}}_0 = 2\ket{0}\!\!\bra{0}^{\otimes 2(n+k)+b(c+1)+1} - \mathbb{I}$. 
    Finally, after $\xi$ such rounds, the qubits are measured. 
    The measurement of the $k$ bottom qubits, corresponding to the eigenvector registers, extracts the information about the relevant eigenvector $\ket{\psi_i} = \ket{\alpha_i}_k\otimes\ket{\phi_i}_n$ of matrix $\widetilde{\mathbf M}$, as defined in Eq.~\eqref{eq:firstMtylda}.
    Therefore, this determines the value $\alpha_i$ that yields matrices $\mathbf{M}(\alpha_i)$ with eigenvalues in the target region.
    }
    \label{fig:circuit-A}
\end{figure}
\end{center}
\twocolumngrid

Subsequently, we initialize $c+1$ copies of $b$-qubit registers in the state $\ket{0}_{b}$.
The number of registers $b$ is related to the desired accuracy, $b = \mathcal{O}\big(\log (1/\varepsilon)\big)$.
These registers will be used for computing candidate eigenvalues of $\widetilde{\mathbf M}$, as shown in Figure~\ref{fig:circuit-A}. The top $b$-qubit register in  Figure~\ref{fig:circuit-A} keeps the final selected candidate eigenvalue $\ket{\tilde{\lambda}}_{b}$, whereas the remaining $c$ clock registers are utilized to store output from a sequence of QPE circuits, i.e., candidate eigenvalues $\eta_1,...,\eta_c$~\footnote{More specifically, for each register, with the highest probability, the outcome will be the true eigenvalue $\lambda$. However, there will also be a small irrelevant part $\ket g$, that will eventually not contribute much to the final estimation of the eigenvalue $\ket\eta = \sqrt{1-\gamma}\ket{\lambda} + \gamma \sum_i \ket{g_i}$, as the coefficient $\gamma \approx 0$.}.
The purpose of the $c$-fold QPE execution is to improve the probability of measuring correct eigenvalues. 
To achieve a target failure probability $\omega$, it suffices to choose $c = \mathcal{O}\big(\log (1/\omega)\big)$, as discussed in App.~\ref{app:accuracy}.
The state representing the median eigenvalue $\ket{\tilde{\lambda}}_{b}$ is constructed from the composite of candidate eigenvalue states $\ket{\eta}_{b}$ generated by each of the QPE circuits. Quantum circuit $\hat{\mathcal{M}}$ depicted in Figure~\ref{fig:circuit-A}  returns the median $\tilde{\lambda}$ of the set $\lbrace\eta_1,...,\eta_c\rbrace$.

The median circuit $\hat{\mathcal{M}}$ can be implemented in a number of ways, including Grover search, the cost of which is negligible~\cite{Kerzner_2024}.  
Choosing the majority voting technique implementing an approximate median, the additional gate costs would be linear in the number of bits $b$ and the number of registers $c$, $\mathcal{O}(bc)$. Thus, the asymptotic scaling of the final circuit is not affected. Generally, following the discussion given in Ref.~\cite{Kerzner_2024}, the number of clock registers is $c=\mathcal{O}(\log \frac{1}{p})$ where $p$ is the probability of not measuring the correct eigenvalue in QPE. 
Note that QPE circuits shown in Figure~\ref{fig:circuit-A} need only to act on one of the states in the pair $\ket{\psi_i}\ket{\bar\psi_i}$ in order to produce an estimate for the eigenvalue of $\widetilde{\mathbf M}$. We denote with $\hat{\mathcal F}$ the quantum circuit for the state preparation given in Eq.~\eqref{eq:init_state}, the successive use of $c$ QPE circuits $\hat{Q}$ and the median circuit $\hat{\mathcal{M}}$ displayed in Figure~\ref{fig:circuit-A}(b), and call $\hat{\mathcal F}$ \textit{eigenvalue finder} circuit. 

For executing each QPE circuit, one needs to implement the evolution unitary $U_{\widetilde{\mathbf M}}(t) = e^{it\widetilde{\mathbf M}}$, which requires $\mathcal O (||\widetilde{\mathbf M}||t+\log(1/\varepsilon))$ calls to block-encoded matrix $\widetilde{\mathbf M}$~\cite{Low_2017}. The number of qubits $b$ in each QPE estimation register shown in Figure~\ref{fig:circuit-A} is determined by target eigenvalue precision $b = \log(1/\varepsilon)$. We adopt the $d$-sparse model for block-encoding individual $\mathbf M_j$ matrices and separately the $\mathbf D^j(\alpha)$ matrices, denoted $\mathcal{B}[\mathbf M_j]$ and $\mathcal{B}[\mathbf D^j]$, respectively. Each block-encoding is associated with an appropriate scaling constant as follows:
\begin{equation}
    \left(\bra{0}_a\otimes \mathbb{I}\right)\mathcal{B}[\mathbf{A}]\left(\ket{0}_a\otimes \mathbb{I}\right) = \frac{\mathbf{A}}{\zeta_A},
\end{equation}
where index $a$ denotes an ancilla register and $\zeta_A$ is a number dependent on the block encoding method. For the $d$-sparse model, $\zeta_A= d||\mathbf{A}||_\textrm{max}$, where $d$ is the maximum number of non-zero elements in any given row of $\mathbf{A}$ and $||\mathbf{A}||_\textrm{max}$ is the maximum-norm of $\mathbf{A}$. Here we assume that matrix $\widetilde{\mathbf M}$ is sparse, i.e., it contains $d=\mathcal{O}(\log(N))$ non-zero entries in each row and column. Note that the presence of energy grid matrices $\mathbf D(\alpha)$ given in Eq.~\eqref{eq:firstMtylda} does not change the matrix sparsity. The detailed costs of the block-encoding circuit can be found, e.g., in Ref.~\cite{camps}. Here we remain agnostic to the block-encoding circuit construction details, only assuming the $d$-sparse scheme for block encoding individual terms $\mathbf D^j(\alpha)$ and  $\mathbf M_j$, while linear combination of unitaries is assumed for block-encoding the sum of block-encoded products $ \mathbf D^j(\alpha) \otimes M_j$. Then the overall block-encoding scaling constant for $\widetilde{\mathbf M}$ is given by $\zeta=\sum_{j=0}^J \zeta(D^j)\zeta(M_j)$, which is upper bounded by $J \max_j\lbrace\zeta(D^j)\zeta(M_j)\rbrace$.
Block-encoding circuit of $\widetilde{\mathbf M}$ requires  $\mathcal{O}(\log (NK)+\log J)$ qubits using the combination of the $d$-sparse and LCU encoding methods~\cite{camps}. 
For further discussion of block-encodings, we refer the reader to Sec.~\ref{subsec:algorithms_complexities}.

The input statevector entering the QPE sequence in eigenvalue \textit{finder} $\hat{\mathcal{F}}$ circuit shown in Figure~\ref{fig:circuit-A}(b) can be written as:
\begin{equation}
    \ket{\Psi}\ket{0}_{(c+1)b}\ket{0}_1.
    \label{eq:gamma0}
\end{equation}
Upon a successful execution of the \textit{finder} circuit, the output state encoding eigenvalues of $\widetilde{\mathbf M}$ can be written in the form:
\begin{widetext}
\begin{equation}
      \ket{\Psi}\ket{0}_{(c+1)b}\ket{0}_1
\xrightarrow{\hat{\mathcal{F}}}  \ket{\Gamma^{(0)}}=\frac{1}{\sqrt{NK}}\sum_{l=0}^{NK-1}\ket{\alpha_l}\ket{\phi_l}\ket{\bar{\alpha}_l}\ket{\bar{\phi}_l}\ket{\eta_l^1}_b\cdots\ket{\eta_l^c}_b\ket{\tilde{\lambda}_l}_{b}\ket{0}_1,
    \label{eq:gamma-1}
\end{equation}
\end{widetext}
where each $\ket{\eta_l}$ corresponds to a superposition over states close to the true eigenvalue $\lambda_l$.
From here on, for clarity, we drop the complex conjugated copy of $\ket{\alpha_l}\ket{\phi_l}$ as well as all the eigenvalue registers apart $\ket{\eta_l^1}_b\cdots\ket{\eta_l^c}_b$ from the median one $\ket{\tilde{\lambda}_l}_{b}$, as defined in Eq.~\eqref{eq:gamma-1}. Our next aim is to increase the amplitudes of all states $\tilde{\lambda}_l$ that satisfy the inequality:
\begin{equation}
    |\tilde{\lambda}_l-\lambda_0|< \varepsilon
    \label{eq:inequality}
\end{equation}
For this purpose, we construct an oracle that marks states satisfying Eq.~\eqref{eq:inequality} by transforming an additional qubit register to state $\ket{1}$ when Eq.~\eqref{eq:inequality} is satisfied and $\ket{0}$ otherwise.  We thus defined the characteristic function  $\chi(x)\in\{0,1\}$ and $\chi(x)=1$ iff $|x-\lambda_0|<\varepsilon$ and will call the associated qubit the \textit{characteristic function qubit}, already shown in Figure~\ref{fig:circuit-A} and Eq.~\eqref{eq:gamma0}.
Let us denote all states satisfying  Eq.~\eqref{eq:inequality} as $\ket{\chi_+} = \sum_{l\in \Omega_{\varepsilon}}\ket{\alpha_l}\ket{\phi_l}\ket{\tilde{\lambda}_l}_{b}\ket{\chi(\tilde{\lambda}_l)}_1$, where $\Omega_{\varepsilon} = \lbrace l: \widetilde{\mathbf M}(\alpha)\ket{\alpha_l}\ket{\phi_l}=\lambda_l\ket{\alpha_l}\ket{\phi_l}  \,\,\text{and}\,\,\,     |\lambda_l-\lambda_0|< \varepsilon \rbrace $~\footnote{In the above, we disregard the registers related to the QPE, as these do not change the value of the characteristic function qubit.}.
For amplifying the amplitude of state $\ket{\chi_+}$, we construct oracle $\hat{R}$ and diffusion operator $\hat{D}$ in the following way:
\begin{align}
    \hat{R} &= \mathbb{I}-2\ket{\chi_+}\!\!\bra{\chi_+} = \mathbb{I}\otimes\mathbb{I}\otimes\cdots \otimes \mathbb{I} \otimes Z,\\
    \hat{D} &= 2\ket{\Psi}\!\!\bra{\Psi}- \mathbb{I} = \hat{\mathcal{F}}^\dagger \hat{\mathcal{R}}_0  \hat{\mathcal{F}}.
    \label{eq:grover}
\end{align}
From the above, we form the Grover iterate operator $\hat{G}= \hat{D}\hat{R}$, which is applied $\xi$-times. 
In the above, the oracle acts as $Z$ operator on the characteristic function qubit, while $\hat{\mathcal{R}}_0 = \ket{0}\!\!\bra{0} - \mathbb{I}$ is a reflection in the computational basis.
Therefore, the action of the finder operator $\hat{\mathcal{F}}$ can be thought of as a change of the basis.

After the $k$-th amplitude amplification iteration, the state of the system can be written as:
\begin{widetext}
\begin{equation}
\begin{split}
    \ket{\Gamma^{(k)}} &= c^{(k)}\left(\sum_{l\in \Omega_{\varepsilon}}\ket{\alpha_l}\ket{\tilde{\lambda}_l}\ket{r_l}\ket{1}\right) + \sqrt{1-(c^{(k)})^2}\left(\sum_{l\notin = \Omega_{\varepsilon}}\ket{\alpha_l}\ket{\tilde{\lambda}_l}\ket{r_l}\ket{0}\right) \equiv   c^{(k)}\ket{\chi_+} +  \sqrt{1-(c^{(k)})^2}\ket{\chi_-},
\end{split}
\end{equation}
\end{widetext}
where $\ket{r_l}$ represents the state of all remaining qubits~\footnote{In principle, due to the not-exact accuracy of QPE, the median qubit will also contain contributions from other states, but we drop this for clarity.}.
Since we do not know the number of eigenvalues satisfying Eq.~\eqref{eq:inequality}, we do not know the amplitude of the \textit{good} state $\ket{\chi}_+$. This means that the straightforward strategy of applying $\big\lfloor\frac{\pi}{4}\sqrt{NK}\big\rfloor$ times the Grover rotation might ``overamplify'' the desired components, leading to a non-optimal amplitude.
For this reason, we follow a modified probabilistic strategy, where in each round we choose a number $r\in [0,1]$. 
Then, we apply Grover rotation $\big\lfloor r \frac{\pi}{4}\sqrt{NK}\big\rfloor$ times. 
The measurement of the characteristic function qubit yields the correct value $\ket{1}_\chi$ with probability
\begin{equation}
\begin{split}
    P(r,N,K,m)&=\sin^2 \bigg( 2 \arcsin{\Big(\sqrt{\frac{m}{NK}}\Big)} \Big\lfloor r \frac{\pi}{4}\sqrt{NK}\Big\rfloor \bigg) \\ &\approx \sin^2\Big( \frac{\pi r\sqrt{m}}{2} \Big),
\end{split}
\end{equation}
where $2 \arcsin{\sqrt{\frac{m}{NK}}}$ is the angle of a single Grover rotation and $m=|\Omega_{\varepsilon}|$. The average probability of success in this case is given by  $P=\int_{0}^1 \text{d}r\sin^2\Big( \frac{\pi r\sqrt{m}}{2} \Big) \geq \frac{1}{2}-\frac{\sin \left(\sqrt{6} \pi \right)}{2 \sqrt{6} \pi } \approx 0.435$, so it is non-vanishing for all values of $m$.
Therefore, applying the above procedure a constant number of times $\mathcal{O}(1)$ allows to obtain a correct value $\alpha_l$ (i.e., $l\in\Omega_{\varepsilon}$) for any threshold probability $p<1$. Thus, our strategy involves performing
$\xi$ iterations of amplitude amplification $\hat{G}$, where $\xi$ is random number from the set $\{1,\ldots,\lfloor\frac{\pi}{4}\sqrt{NK} \rfloor\}$. 

Alternatively, instead of using the above-mentioned randomized strategy, one could utilize a more involved approach like the optimal fixed-point method~\cite{Yoder_2014}, what would similarly remove the overshooting problem.
As an advantage, this would give monotonic convergence and explicit failure-probability control.
For simplicity, we do not discuss this alternative here as the asymptotic scalings are the same.
A more detailed discussion of the probabilities of success of this scheme is provided in App.~\ref{app:accuracy}.

In the end, we are interested in values of parameter $\alpha$ for which an eigenvalue $\lambda$ of $\widetilde{\mathbf M}(\alpha)$ exists satisfying Eq.~\eqref{eq:inequality}. Measuring register $\ket{\alpha}_{k}$ returns a bitstring.
For the case of $J=1$ it might represent eigenvalues $\alpha_k$ of the generalized eigenvalue problem, associated for example with the discretized Schr\"odinger equation, as discussed in Sec.~\ref{subsec:collocation_method}. 

In summary, the overall number of QPE calls for finding an eigenvalue $\lambda$ that satisfies Eq.~\eqref{eq:inequality} with probability $p$ scales as $\mathcal{O}(\sqrt{KN})$.

    \end{proof}

In a special case of a single matrix $\mathbf M(\alpha')$ out of $K$ matrices $\{\mathbf M(\alpha_i)\}_{i=1}^K$, having an eigenvalue $\lambda$ satisfying conditions given in Eq.~\eqref{eq:inequality}, our algorithm returns product state
   \begin{equation}
       \ket{\alpha'}_k \ket{\lambda'}_{b}\ket{r'}\ket{1},
   \end{equation}
and upon measuring the $k=\log K$ qubit register $\ket{\alpha'}_k$ returns the relevant value $\alpha'$ with probability close to 1.
In case there is more than one matrix with eigenvalues in the corresponding range, the resulting state can be written as
    \begin{equation}
        \sum_{j\in A}\ket{\alpha_j} \ket{\lambda_j}_{b}\ket{r_j}\ket{1},
    \end{equation}
    where we defined $A = \lbrace j: \exists \lambda: \;  \widetilde{\mathbf M}(\alpha_j)\ket{\alpha_j}\ket{\phi_j}=\lambda_{j}\ket{\alpha_j}\ket{\phi_j} \wedge    |\lambda_j-\lambda_0|< \varepsilon \rbrace$ and the probability of measuring any $\ket{\alpha_j},j\in A,$ is uniform. In cases when more than one matrix eigenvalue satisfies Eq.~\eqref{eq:inequality}, the output state can be written as
    \begin{equation}
        \sum_{j\in A}\gamma_j\ket{\alpha_j} \ket{\lambda_j}_{b}\ket{r_j}\ket{1},
        \label{eq:quasidegeneracy}
    \end{equation}
    where $\gamma_j$ is proportional to the number of eigenvalues $\lambda$ (quasi-degeneracy) for a given $\alpha_j$. Such a scenario is plausible in solid-state physics, where the density of states for Hamiltonians with quasi-continuous (gapless) spectra is of interest. With a histogram estimate of the density of states function, one can predict several useful properties of materials, such as conductivity, heat capacity, magnetic susceptibility, and absorption spectra, or characterize superconducting systems within the BCS theory, modeling tunneling rates~\cite{cardona2005,kittel,nieminen2002}. In case $\epsilon$ is chosen such that $\gamma_k=\gamma_{k'}$ in Eq.~\eqref{eq:quasidegeneracy}, i.e., there is no quasidegeneracy, all correct values $\alpha_k$ can be obtained in $\mathcal{O}(L\log L)$ queries, where $L$ is the size of $A$ (the number of ``correct'' indices), which follows from the coupon collector's problem.
    
    A schematic depiction of eigenvalues of $\widetilde{\mathbf M}(\alpha)$ is given in Fig.~\ref{fig:landscape_sketch}. In this figure, for all $\alpha \in A$, matrices $\mathbf M(\alpha)$ satisfy Eq.~\eqref{eq:inequality}. Note that, when there are no matrices whose eigenvalues satisfy Eq.~\eqref{eq:inequality}, the probability for measuring any given $\alpha_k$ is uniform across the grid. In such a scenario, a finer grid for $\ket{\alpha_k}$ might be required or the value of $\lambda_0$ and $\varepsilon$ should be adjusted.

 Importantly, the probability of sampling a given index $\alpha$ is uniform over all  $\alpha_i$ and $\alpha_j$ for $i,j\in \Omega_{\varepsilon}$,
    \begin{equation}
        \lim_{K,1/\varepsilon \to \infty}  \frac{p_i}{p_j} = 1,
    \end{equation}
    in the limit of large number of points $K$ and high accuracy $1/\varepsilon$, while their product is kept constant, $K\varepsilon = \text{const}$.
    In particular, in this limit, the probability of sampling from a given region is inversely proportional to the derivative in the said region,
    \begin{equation}
        \lim_{K,1/\varepsilon \to \infty} \frac{p(\alpha\in A_i)}{p(\alpha\in A_j)} = \frac{\lambda'_{l}(x_j)}{\lambda'_{l}(x_i)}.
    \end{equation}
    where $\lambda'_{l}(x_i)$ denotes the derivative of the selected eigenvalue (singular value) $\lambda_{l}$ with respect to parameter $\alpha$, taken in the central point $x_i$ of the region $A_i$.

\textbf{Rectangular matrices.}
A generalization of Theorem~\ref{thm:quantum_algorithm_square} for families of rectangular matrices  $\mathbf M(\alpha)$ can be achieved by embedding them into larger square matrices (augmented matrices).
  Similarly to the case of Hermitian matrices, we can write the following corollary:  
    \begin{corollary}\label{corr:quantum_algorithm_rectangular}
        Consider a discretized one-parameter family of rectangular matrices $\mathbf M(\alpha)$ on a grid formed by $K$ points. 
        The larger of the matrix dimensions is denoted by $N$.
        If the dependence on $\alpha$ can be written as a finite series expansion $\sum_{j=0}^J \alpha^j \mathbf M_j$, for some matrices $\mathbf M_j$, there exists a quantum algorithm that finds a single instance of the parameters $\alpha'$ such that matrices $\mathbf M(\alpha')$ have singular values in region $[\sigma_0-\varepsilon,\sigma_0+\varepsilon]$ with $\widetilde{\mathcal{O}}(J\zeta\sqrt{NK}/\varepsilon)$ calls to block-encoded matrices $\mathbf M_j$, which all have their max-norms upper-bounded by $\zeta$.
    \end{corollary} 

    \begin{proof}
    In the general rectangular case, the matrices $\mathbf{M}_j$ of size $N\times M$ can be tensor-multiplied by $\mathbf{D}^j(\alpha)$, analogously to the proof for Theorem~\ref{thm:quantum_algorithm_square}. 
    Here, we apply a similar reasoning as in the proof for Theorem~\ref{thm:quantum_algorithm_square}, assuming an augmented, square matrix of order $K(N+M)$
    \begin{widetext}
    \begin{equation}\label{secondMtylda}
    \begin{split}
        \widetilde{\mathbf M} &= \ket{1}\bra{0}\otimes \Big( \sum_j \mathbf D^j(\alpha) \otimes \mathbf M_j\Big)+\ket{0}\bra{1}\otimes \Big(\sum_j \mathbf D^j(\alpha) \otimes \mathbf M_j\Big)^\dagger =  \begin{pmatrix}
            0 & \sum_j \mathbf D^j(\alpha) \otimes \mathbf M_j^\dagger \\
            \sum_j \mathbf D^j(\alpha) \otimes \mathbf M_j & 0
        \end{pmatrix}.
    \end{split}
    \end{equation}
    \end{widetext}
    in place of matrix $\widetilde{M}$ given in~\eqref{eq:firstMtylda}, with certain differences.
    
     Nonnegative eigenvalues of $\widetilde{\mathbf M}$ are exactly the singular values of $ \sum_j \mathbf D^j(\alpha) \otimes \mathbf M_j$.
    Matrix $\widetilde{\mathbf M}$ has $2MK$ eigenvectors that can be written in the form $\ket{g_{\pm k}}=\ket{0}\ket{e_k}\pm\ket{1}\ket{f_k}$, with eigenvalues $\pm\sigma_k$, as well as $K(N-M)$-dimensional subspace orthogonal to them, contained in the kernel. 
     Vectors $\ket{e_k}$ and $\ket{f_k}$ are, respectively, the right and left singular vectors of matrix $ \sum_j \mathbf D^j(\alpha) \otimes \mathbf M_j$ corresponding to the same singular value $\sigma_k$.
     Subsequently, we prepare our initial state in a slightly different way, as compared to the square Hermitian case, namely:
\begin{equation}
\begin{split}
\ket{\Psi}&=\sum_{k=1}^{N}\ket{0}\ket{k}\ket{0}\ket{k} = \sum_{k=1}^{N}\ket{0}\ket{e_k}\ket{0}\ket{\bar{e}_k} \\ &=\sum_{k=1}^{N}\left(\ket{g_{+k}}+\ket{g_{-k}}\right)\ket{0}\ket{\bar{e_k}}.
\end{split}
\end{equation}

 Similarly to the proof of Thm.~\ref{thm:quantum_algorithm_square}, we apply QPE with the median trick and the region oracle.
 Upon application of the \textit{finder circuit}, defined in Eq.~\eqref{eq:grover}, the statevector can be written as 
\begin{equation}
\begin{split}
    &\ket{0}\ket{0}_b\ket{\Psi}\,\, \xmapsto{\mathrm{QPE},\,\, \mathrm{region\,\, oracle}} \,\, \ket{\Phi}, \,\,\text{ where} \\ \ket{\Phi}=&\sum_{k=1}^{N}\ket{\chi(\sigma_{k})}_\chi\left(\ket{\sigma_k}_b\ket{g_{k}}+\ket{-\sigma_i}_b\ket{g_{-k}}\right)\ket{0}\ket{\bar{e}_k},
\end{split}
\end{equation}
with $\chi(\sigma_k)=1$ iff $|\sigma_k - \sigma_0|<\varepsilon$ and we dropped for clarity all $c$ clock registers storing subsequent candidate eigenvalues.  
Combining the costs of state$\ket{\Psi}$ preparation, i.e.,   $\Od{\log{(KN)}\log(1/\varepsilon)/\varepsilon}$ gives the total cost
\begin{equation}
    \Od{J\log{(KN)}\sqrt{KN}\log(1/\varepsilon)/\varepsilon},
\end{equation}
for estimation of $\sigma_i$ with accuracy $\varepsilon$. Dropping logarithmic terms retrieves scaling expressed in the $\widetilde{\mathcal O}$ notation of Corrolary~\ref{corr:quantum_algorithm_rectangular}.

    \end{proof}
    
    The remainder of the paper is devoted to a study of a particular application of the present method, namely solving the Schr\"odinger equation, in which case $J=1$. The most general case then involves rectangular matrix representations. 

\section{Applications: finding spectrum of a Hamiltonian}\label{sec:application}

\subsection{Collocation method}\label{subsec:collocation_approximation}
We consider a specific case of Eq.~\eqref{eq:firstMtylda} for $J=1$, which produces the generalized eigenvalue problem as presented in Eq.~\eqref{eq:gen_eval}. In quantum-mechanical calculations, such a generalized eigenvalue problem arises naturally from the discretization of the Schr\"odinger equation, either via projection onto a basis set in $L_2$ (Galerkin method~\cite{Helgaker2000}) or onto Dirac delta distributions located at grid points (collocation method~\cite{Boys1969,Peet1989,Carrington2017}). Both collocation and Galerkin methods have proven effective in electronic structure calculations~\cite{Carrington2017,Manzhos2023,Helgaker2000}, as well as in rovibrational molecular simulations~\cite{Avila2015,Brown2015}.
When a non-orthogonal basis is used, or when the Gram matrix in the Galerkin-discretized Schr\"odinger equation is evaluated using inexact quadratures, the result is a generalized matrix eigenvalue problem. In the collocation method, the generalized eigenvalue structure arises by construction~\cite{Peet1989,Avila2015,Manzhos2023}. 

To compare with our proposed quantum algorithm discussed in Sec.~\ref{sec:general_method} for solving generalized eigenvalue problems, we now derive and briefly discuss key aspects of the collocation technique. In the following sections, we compare the computational complexity of classical approaches to solving collocation equations with our quantum computing proposal, highlighting advantages and disadvantages. 

In the collocation method, the Schr\"odinger equation for the Hamiltonian, written as
\begin{equation}\label{eq:collocation_SE}
\hat{H} = \hat{K} + \hat{V},
\end{equation}
is discretized by expanding the solution in a basis set $\lbrace\ket{\phi_n}\rbrace_{n = 0,\dots,N-1}$ of functions that are not necessarily orthonormal, and enforcing the solution solves the Schr\"odinger equation exactly, for a chosen grid of spatial points represented by vectors $\lbrace\ket{q_k}\rbrace_{k=0,1,\dots,M-1}$. 
As a result of such a representation choice, the  following matrix eigenvalue equation can be formed~\cite{Peet1989,Carrington2017}:
\begin{equation}
\mathbf{H}\mathbf{U} = \mathbf{B}\mathbf{U}\mathbf{E}.
\label{eq:collocation}
\end{equation}
where $\mathbf{U}$ is the matrix formed by the Hamiltonian’s eigenvectors, $\mathbf{E}$ is the diagonal matrix of eigenvalues, $\mathbf{E} = \text{diag}(E_0,\dots,E_D)$ and $\mathbf B_{kn} = \phi_n(q_k)$ is the \textit{collocation matrix} keeping values of basis wavefunctions evaluated at grid points.

Matrix elements of the collocation Hamiltonian can be written as:
\begin{equation}
\mathbf{H}_{kj} = \langle q_k |\hat{H}| \phi_j \rangle = \left(\mathbf{K}\mathbf{B}\right)_{kj} + (\mathbf{V}^{\text{diag}}\mathbf{B})_{kj}
\label{eq:collocation2}
\end{equation}
where $\mathbf {K}$ is the second-derivative kinetic energy operator,  and $\mathbf{V}^{\text{diag}}$ is the diagonal potential energy operator, with values of the potential at the grid points: $\mathbf{V}^{\text{diag}}_{kn} = \delta_{kn} V(q_k)$. If the number of basis functions equals the number of grid points ($N = M$), Eq.~\eqref{eq:collocation} becomes a generalized matrix eigenvalue problem for a square, non-symmetric matrix.
Eq.~\eqref{eq:collocation2} can be further simplified by noting that the collocation approximation is formally equivalent to approximating the representation of the second-derivative operators appearing in the kinetic energy operator~\cite{Boys1969} as follows:
\begin{equation}
\begin{split}
\mathbf K_{nm} &= -\bra{\phi_n}\frac{\partial^2}{\partial q^2}\ket{\phi_m} = -\int \text{d}q \phi_n^*(q) \phi''_m(q) \\ &\approx -\sum_{k \in \text{grid}} \phi_{n}^*(q_k) \phi''_m(q_k) = [\mathbf B^{\dagger}\mathbf B'']_{nm}.
\end{split}
\end{equation}
as originally proposed by Boys~\cite{Boys1969}.

The collocation Schr\"odinger equation~\eqref{eq:collocation} can be formulated with general kinetic energy operators expressed in curvilinear coordinates 
\begin{equation}
    \hat{K}=\sum_{l,l'=1}^{D}G_{ll'}(\mathbf{q})\frac{\partial^2}{\partial q_l\partial q_{l'}}+\sum_{l=1}^{D}F_{l}(\mathbf{q})\frac{\partial}{\partial q_l},
\end{equation}
where $G_{ll'}(\mathbf{q})$ and $F_{l}(\mathbf{q})$ represent coordinate-dependent matrices and $D$ is the number of coordinates (for molecules $D=3N_{at.}-6$, where $N_{at.}$ is the number of atoms). The potential energy surfaces can adopt a general form giving the following representation of the collocation equations~\cite{Avila2015,Carrington2021}:
\begin{equation}\label{eq:initial_schroedinger}
(\mathbf{B}'' + \mathbf{V}^{\text{diag}}\mathbf{B})\mathbf{U} = \mathbf{B}\mathbf{U}\mathbf{E},
\end{equation}
where
\begin{equation}
\mathbf B''_{kn} = (\hat{K}\phi_n)(q)|_{q=q_k}
\label{eq:second-derivative}
\end{equation}
represents the action of the kinetic energy operator on the $n$-th basis function evaluated at the grid point $q_k$. For the one dimensional harmonic oscillator model, the kinetic energy operator in the normal coordinates reads $\hat K = -\frac{\partial^2}{\partial q^2}$, such that the \textit{second derivative} matrix defined in Eq.~\eqref{eq:second-derivative} can be written as $\mathbf B''_{kn}=\phi''_n(q_k)$.  We further consider the form of collocation equations given in Eq.~\eqref{eq:initial_schroedinger}, noting that its form is general, not limited to the harmonic approximation nor a specific potential or kinetic energy operator. 

\subsection{Solving collocation equations: matrix inverse}\label{subsec:collocation_method}
Collocation equations given in Eq.~\eqref{eq:initial_schroedinger} can be solved by reformulating the problem as a regular square eigenvalue problem. For this purpose, $\mathbf B^{\dagger}$ can be applied to both sides of Eq.~\eqref{eq:initial_schroedinger} to obtain
\begin{equation}\label{eq:classical_collocation}
(\mathbf B^{\dagger}\mathbf B'' + \mathbf B^{\dagger} \mathbf V \mathbf B)\mathbf U = \mathbf B^{\dagger} \mathbf B \mathbf U \mathbf E,
\end{equation}
where $\mathbf B^{\dagger} \mathbf B$ is, in general, a non-diagonal symmetric (invertible) matrix. By acting on both sides of Eq.~\eqref{eq:classical_collocation} with the inverse of $\mathbf B^{\dagger} \mathbf B$, one obtains a regular non-symmetric eigenvalue problem written as

\begin{equation}\label{eq:collocation_SE_simplified}
\underbrace{(\mathbf B^{\dagger} \mathbf B)^{-1}(\mathbf B^{\dagger}\mathbf B'' + \mathbf B^{\dagger}\mathbf V\mathbf B)}_{\widetilde{\mathbf H}} \mathbf U = \widetilde{\mathbf H}\mathbf U = \mathbf U\mathbf E,
\end{equation}
which can be solved by the Arnoldi procedure~\cite{Lehoucq1998} or more recent alternatives~\cite{Myllykoski_2020}. Matrix inversion introduces numerical stability issues when the condition number $\kappa$ of the $\mathbf B^{\dagger} \mathbf B$ matrix is large. For this reason, and due to the $\mathcal{O}(N^3)$ computational cost of matrix inversion, this solution technique is not always feasible.

Similarly, a quantum computer implementation, potentially involving the HHL algorithm~\cite{Harrow2009} for matrix inversion followed by QPE, faces several challenges. Despite the generally favorable scaling of the HHL matrix inversion algorithm compared to exact classical algorithms (i.e.,  $\mathcal{O}(\log(N) d^2 \kappa^2 / \epsilon)$ vs.\ $\mathcal{O}(N^3)$), where $d$ is matrix sparsity, the HHL algorithm is expected to perform poorly for problems with high condition numbers. The dependence on condition number in classical iterative algorithms, such as the Arnoldi iteration, is more favourable, being $\mathcal{O}(\kappa^{\frac{1}{2}})$. It may therefore be more advantageous to design the basis set such that matrix inversion can be performed classically, which is straightforward for direct product basis sets and pruned non-direct basis sets~\cite{Zak2019}.

Secondly, when a regular eigenvalue problem is formed, one could attempt to use Quantum Phase Estimation to find eigenvalues associated with Eq.~\eqref{eq:collocation_SE_simplified}, but this entails the challenge of sampling enough eigenvalues of interest. Quantum algorithms for Hamiltonian simulation based on QPE~\cite{babbush2018,lee2021,su2021,trenev2025} are inherently designed to solve the ground state problem or compute a few lowest eigenvalues. When many eigenvalues are required, such as in solving the ro-vibrational molecular problem, an alternative approach is needed.

In the next section, we describe a procedure for solving collocation equations without matrix inversion, designed for parallel computation of many eigenvalues on either classical or quantum computers.  

\subsection{Solving collocation equations: landscape scanning with classical computers}\label{sec:landscape_scanning}
The collocation equation
\begin{equation}\label{eq:initial_schroedinger2}
(\mathbf{B}'' + \mathbf{V^{diag}}\mathbf{B})\mathbf{U} = \mathbf{B}\mathbf{UE}
\end{equation}
can be solved by first constructing the matrices $\mathbf{B}, \mathbf{B}''$ and $\mathbf{V}^{diag}$, and computing the lowest singular value of the residue operator defined as:
\begin{equation}
\mathbf M(\alpha) \coloneqq \mathbf{B}'' + \mathbf{V^{diag}}\mathbf{B} - \alpha \mathbf{B}
\label{eq:residue}
\end{equation}
where $\alpha$ is a real parameter (cf. Eq.~\eqref{eq:firstMtylda}). When $\alpha$ equals one of the singular values $E_i$ satisfying Eq.~\eqref{eq:initial_schroedinger2}, the following equation is satisfied:
\begin{equation}\label{eq:landscape_SE_simplified}
(\mathbf{B}'' + \mathbf{V^{diag}}\mathbf{B} - E_i \mathbf{B})u_i = 0.
\end{equation}
where $u_i$ is a singular vector corresponding to the singular value $E_i$.

In practice, it is sufficient to find such an $\alpha$ that the singular values of $\mathbf M(\alpha)$ are below some threshold value $\varepsilon$. The quality of the basis set used to represent the Schr\"odinger equation written in Eq.~\eqref{eq:initial_schroedinger2} dictates the minimal achievable singular value for the residue matrix given in Eq.~\eqref{eq:residue}. In particular, if the basis set is complete, i.e., the space spanned by the basis functions contains the eigenspace for the Hamiltonian, the residue function has nodes corresponding to eigenvalues of the Hamiltonian. 
We may thus calculate the residue matrix defined in Eq.~\eqref{eq:residue} for a selected grid (energy landscape scan) of candidate eigenvalues $\lbrace \alpha_j\rbrace_{j=1,...,K}$ to determine minima in the singular values, as shown in Figure~\ref{fig:26functions}.

\subsection{Comparing the matrix inversion method and landscape scanning}\label{subsec:limitations}
For cases when the condition number $\kappa$ of  $\mathbf B^{\dagger} \mathbf B$ is low and high precision $\delta$ for eigenvalues $E_i$ is required, it is reasonable to use the matrix inverse method as discussed in Sec.~\ref{subsec:collocation_method}. However, as we discuss below, for distributed Gaussian~\cite{Manzhos2023} basis sets with non-zero overlaps used for representing the collocation equations, the condition number will grow fast as the basis set size increases, reaching intractable values of $10^{16}$-$10^{20}$ already for 30 to 40 basis functions in 1 dimension. In this section, we compare the matrix inverse method with eigenvalue landscape scanning and identify cases for which the latter is advantageous. 
    
First, we study the exactly solvable one-dimensional harmonic oscillator model, with $V(x) = x^2$, in the region $x \in [-10,10]$. Both methods require the same data as input: matrices $\mathbf{B}''$, $\mathbf{V^{diag}}$ and $\mathbf{B}$. For the basis set, we choose $N$ real Gaussian functions with a constant width, centered at $M = 80$ equidistant grid points $x_i$.
Our tests involve increasing the number of basis functions while keeping the number of grid points constant. The results for $N=26$ basis functions are shown in Fig.~\ref{fig:26functions}.
Estimation of low energies (red dashed lines) is accurate up to the ninth excited state (exact value $E_9  = 19$), for which the matrix inverse method gives $E_i=19.062$, marked with solid green line in Fig.~\ref{fig:26functions}. The error in the tenth excited state is already 30\% of the energy spacing. The accuracy in energy levels stops increasing for the matrix inverse method with more than 26 basis functions. 
The choice and number of grid points also play an important role in collocation, although a detailed analysis is beyond the scope of this work, we refer the reader to Refs.~\cite{Manzhos2023,Carrington2021,Avila2015} for comprehensive discussions. In practice, provided that a sufficiently high-quality basis set is used, the precise distribution of grid points is less critical. It is generally observed that the number of grid points must exceed the number of basis functions to achieve convergence, with this requirement closely tied to basis set quality. Increasing the number of grid points typically leads to a higher condition number of the collocation matrix, which can reduce numerical stability and necessitate higher-precision arithmetic. Having said that, the increase in the number of basis functions in one dimension also naturally leads to greater condition numbers.

Notably, when the number of basis functions is increased, the condition number of the $\mathbf B^{\dagger} \mathbf B$  matrix grows and reaches $5\times 10^{16}$ already for 35 basis functions, which poses problems for the matrix inverse method, as depicted in Fig.~\ref{fig:35-36functions}. The red dashed lines corresponding to eigenvalue estimates from the double-precision matrix inverse method estimate reasonably only the first few excited states, while the other eigenvalues are complex (not shown).    
In comparison, the eigenvalue landscape scanning method, which witnessed worse accuracy than the matrix inverse method for the lower number of basis functions (cf.~Fig.~\ref{fig:26functions}), provides advantages for 35 basis functions. The first advantage is that it is insensitive to the condition number and captures eigenvalues for highly excited states, beyond $E_{20}$. Thus, the preferred method for calculating higher excited state energies is landscape scanning. The sensitivity of eigenvalue identification through landscape scanning can be further improved by removing background trend, as explained in App.~\ref{app:background_removal}.
In the following figures, we apply this transformation for an improved clarity.

In summary, when the condition number is not prohibitive, the matrix inversion method yields more accurate energy estimates than landscape scanning. However, because the condition number of the associated matrices grow rapidly with basis set and grid size (not only for distributed Gaussian basis sets), we conclude that for larger systems and when highly excited states are of concern, the landscape scanning method is preferable. 

\onecolumngrid
\begin{center}
 \begin{figure}[h!]
       \centering
       \includegraphics[width=0.9\linewidth]{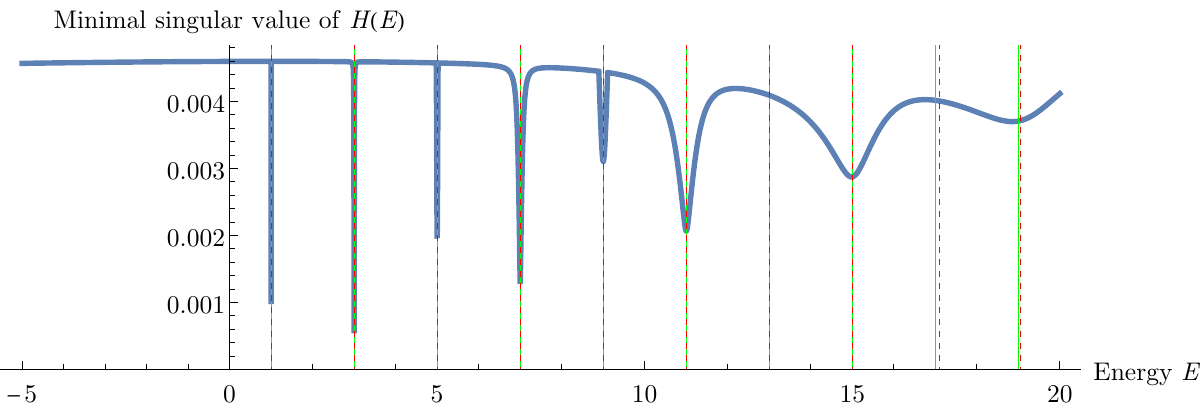}
       \caption{
 Comparison of the matrix inverse method described in Sec.~\ref{subsec:collocation_method} with the classical landscape scanning method for solving the collocation Schr\"odinger equation for the 1D harmonic oscillator potential, using 26 basis functions and 80 grid points. The  condition number of $\mathbf B^\dagger \mathbf B$ equals to $\kappa = 1.53\times 10^{9}$. The exact energies $E_n = 2n+1$ are denoted by solid, vertical green lines. The blue curve corresponds to the minimal singular value produced by the landscape method, while the dashed vertical lines correspond to the eigenvalue estimates obtained by the inversion method.}
       \label{fig:26functions}
   \end{figure}
\end{center}
\twocolumngrid
   
  \onecolumngrid
\begin{center}
   \begin{figure}[h!]
       \centering
       \includegraphics[width=0.9\linewidth]{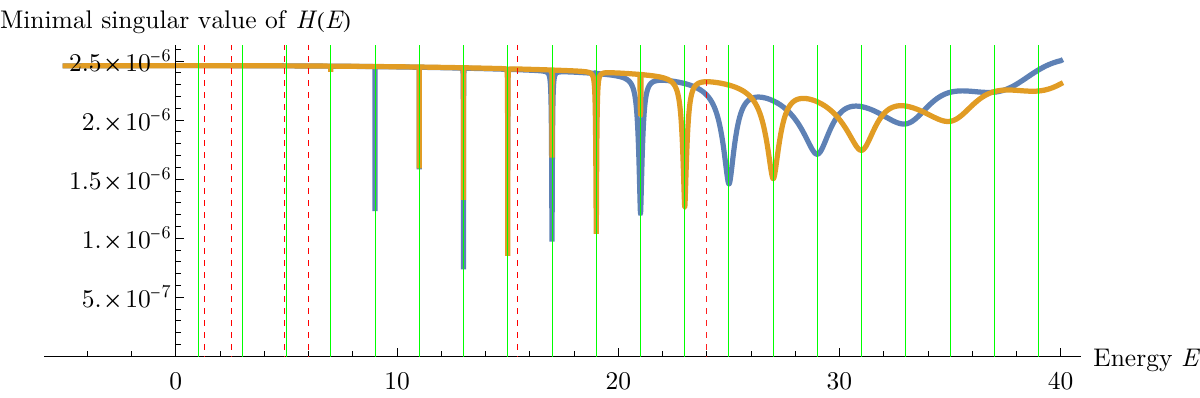}
       \caption{The same setup as in Fig.~\ref{fig:26functions} with 35 (blue) and 36 (orange) basis functions, with 80 grid points. The condition numbers respectively $1.21\times 10^{16}$ and $5.47\times 10^{16}$. Two different sets of basis functions are combined to account for the parity of the solutions to the Schr\"odinger equation.}
       \label{fig:35-36functions}
   \end{figure}
   \end{center}
\twocolumngrid

    Similarly to the harmonic oscillator model, we studied the collocation method with the Morse potential given by the following formula
    \begin{equation}\label{eq:Morse_potential}
        V_M(x) = D_e (1- e^{-ax})^2,
    \end{equation}
    where we have chosen the parameters to be $D_e = 16$ and $a = 4$.
    The basis functions are defined in the same way as for the harmonic oscillator case.
    For low numbers of basis functions, the matrix inverse method is again superior to landscape scanning.   Fig.~\ref{fig:morse_38functions} compares the eigenvalue scanning method with matrix inversion and demonstrates that, to correctly estimate higher excited state energies, an increased number of basis functions is required, giving a high condition number, over $1 \times 10^{16}$. In this case, landscape scanning captures eigenvalues, whereas the matrix inverse method becomes numerically unstable. The number of grid points in this case has less influence on the accuracy of the energy levels.
    
\onecolumngrid
\begin{center}
       \begin{figure}[h!]
       \centering
       \includegraphics[width=0.9\linewidth]{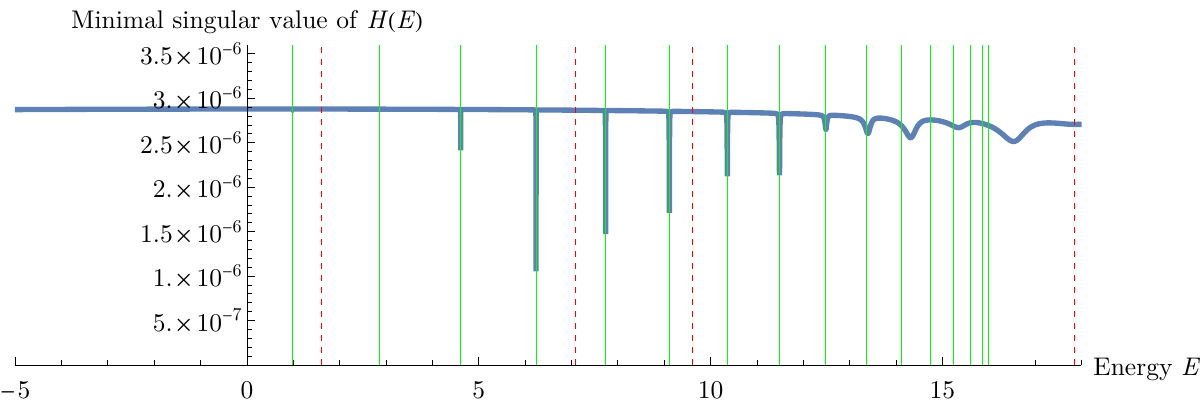}
       \caption{Comparison of the matrix inverse method described in Sec.~\ref{subsec:collocation_method} with the classical landscape scanning method for solving the collocation Schr\"odinger equation for the 1D Morse potential, given by Eq.~\eqref{eq:Morse_potential}, using 35 basis functions  and 80 grid points, with the condition number of $\mathbf B^\dagger \mathbf B$ equal to $1.23 \times 10^{16}$.
The curves have the same meaning as in Fig.~\ref{fig:26functions}.
Note that the matrix inversion method for solving the collocation equations (red dashed line) fails to correctly predict even a single energy level (green solid lines).}
       \label{fig:morse_38functions}
   \end{figure}
   \end{center}
\twocolumngrid

\subsection{Solving collocation equations: landscape scanning with quantum computers}
\label{sec:quantum-landscape}
In this section, we present details of our quantum computing algorithm for energy landscape scanning. For demonstration, we aim at solving the collocation equations given in Eq.~\eqref{eq:initial_schroedinger}, by first constructing the residue operator in an extended Hilbert space as follows:
    \begin{equation}
       \widetilde{\mathbf M}(\alpha) =  \mathbb{I} \otimes (\mathbf{B}'' + \mathbf{V^{diag}}\mathbf{B} ) - \text{diag}(\alpha_1,...,\alpha_K) \otimes \mathbf{B} . 
       \label{eq:collocation_qscan}
    \end{equation}
    where the grid of candidate eigenvalues $\alpha_i$ is encoded in basis states of an additional Hilbert space $\mathcal{H}_G$. Comparing with Eq.~\eqref{eq:firstMtylda} that defines a general method, we set $\mathbf M_0=\mathbf{B}'' + \mathbf{V^{diag}}\mathbf{B}$ and $\mathbf M_1=\mathbf B$. 
    Matrix $\widetilde{\mathbf M}(\alpha)$ is the residue operator that acts on states from the extended Hilbert space $\mathcal{H} = \mathcal{H}_G\otimes \mathcal{H}_{sys} \equiv \mathbb{C}{^{2^K\times 2^K}}\otimes  \mathbb{C}^{2^M\times2^N}$.
    Next,  apply Corollary~\ref{corr:quantum_algorithm_rectangular} to matrix $\widetilde{\mathbf M}(\alpha)$ defined in Eq.~\eqref{eq:collocation_qscan}. A schematic overview of the proposed algorithm is depicted in Fig.~\ref{fig:sketch_quantum_algorithm}.

    Given block-encodings of matrices $\mathbf M_0$ and $\mathbf M_1$, we can compute the energies $E_i$ in block-encoding query complexity $\mathcal{\widetilde O}(\zeta\sqrt{NK}/\varepsilon)$, where $N$ is the number of basis functions, $K$ is the number of energy grid points, and $\zeta$ is the total block-encoding scaling constant, as discussed in Sec.~\ref{sec:general_method}. 

    \onecolumngrid
    \begin{center}
    \begin{figure}[h]
        \centering
        \includegraphics[width=1\linewidth]{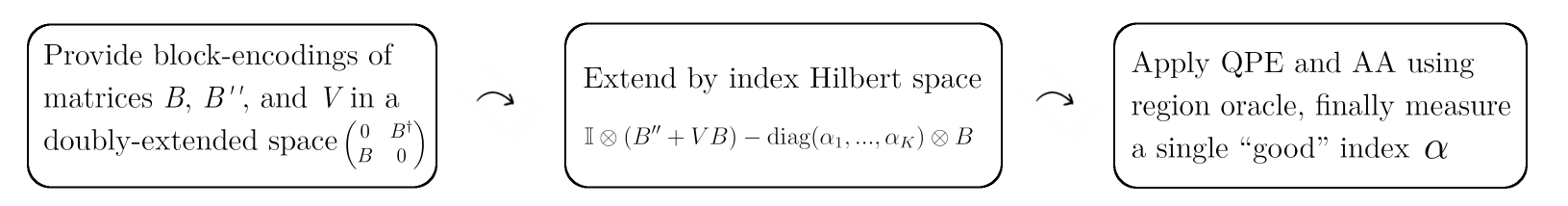}
        \caption{Sketch of the quantum algorithm for searching energies for a given collocation Hamiltonian defined by the collocation matrix $\mathbf \mathbf B$, its second derivative $\mathbf B''$ and the potential $\mathbf V$. Quantum Phase Estimation (QPE) and amplitude amplification (AA) algorithms are used with an oracle determining if eigenvalues belong to a predefined region, as written in Eq.~\eqref{eq:inequality}.
        }
        \label{fig:sketch_quantum_algorithm}
    \end{figure}
    \end{center}
    \twocolumngrid

    \subsection{Comparison of algorithms' complexities}\label{subsec:algorithms_complexities}
    
    We compare the complexities of two classical algorithms for solving the collocation equations and two quantum algorithms. For the classical algorithms, we discuss the matrix inverse method presented in Sec.~\ref{subsec:collocation_method} and landscape scanning methods presented in Sec.~\ref{sec:landscape_scanning}. As a comparison, we provide the solution to a problem of estimating eigenvalues in Eq.~\eqref{eq:collocation_SE_simplified}, with our quantum landscape scanning method based upon QPE.
    
    For comparing quantum and classical algorithms we choose the following parameters: accuracy $\varepsilon$, the number of basis functions $N$~\footnote{To avoid obfuscation by too many parameters, we also assume the number of grid points in the spatial coordinate $M$ is proportional to $N$, i.e.,  $M = \mathcal{O}(N)$, making the ratio of the dimensions of the initial matrices $\mathbf{B}$ and $\mathbf{B}''$ constant. }, the maximal element (max-norm) of matrices $M_\text{max} \coloneqq \text{max}\{||\mathbf{B}||_\textrm{max},||\mathbf{B''}||_\textrm{max},||\mathbf{V}^{diag}||_\textrm{max}\}$, and the number of energy grid points $K$.
    Our comparison also assumes that each algorithm under consideration aims to find at least one eigenvalue of a given Hamiltonian.
    Finally, we assume that the matrices are $d$-sparse, i.e.,  each row and column contains at most $d$ non-zero elements~\footnote{This assumption is valid for localized basis sets, such as the Gaussian basis.}. 
    Under these assumptions, the complexities of all  three cases are shown in Table~\ref{tab:complexity_comparison}.

\onecolumngrid
    \begin{center}
\begin{table}[H]
    \centering
\begin{tabular}{|cc|c|}
\hline
\multicolumn{2}{|c|}{classical}    & \multicolumn{1}{|c|}{quantum} \\ \cline{1-3}
  matrix inversion &    \multicolumn{2}{|c|}{landscape scanning\hspace{0.5cm}\phantom{1}}              \\ \cline{1-3}
\multicolumn{1}{|c|}{$N^{2.371}\log(\kappa)+N^2\kappa^{\frac{1}{2}}\log(1/\varepsilon)$} & $d^2N+KdN\log(1/\varepsilon)$ & $M_\text{max} \frac{\sqrt{K}N^{3/2}d}{\varepsilon}$     \!\!\!\!\!\!\!\!\! \phantom{$\frac{\int}{\int}$}        \\ \hline
\multicolumn{1}{|c|}{condition number problem} & --- & ---     \\ \hline
\end{tabular}
\caption{Computational and (and T-gate) complexities $\mathcal{O}$ for three algorithms solving the Schr\"odinger equation, considered in this work: the matrix inverse method and two landscape scanning methods, from the perspective of both classical and quantum computing.
The relevant parameters are mentioned in the text above. 
From the computational perspective, the key parameter for the matrix inverse method is condition number $\kappa$, causing numerical instabilities at matrix inversion. In particular, the condition number problem limits the number of Gaussian basis functions that can be used per unit volume in the physical space, thus restricting calculations to a few excited states.
The issue with condition number is absent for the landscape method, where scans of singular values of the residue operator $\mathbf M(\alpha)$ are carried out without matrix inversion, see Sec.~\ref{sec:landscape_scanning}.}
\label{tab:complexity_comparison}
\end{table}
\end{center}
\twocolumngrid

We discuss below the algorithmic complexities given in Table~\ref{tab:complexity_comparison}.
Comparing quantum computing to classical computing algorithmic complexity is presently sanctioned by numerous assumptions related to quantum gate time execution, quantum error correcting code overhead and its logical error rate as well as decoding and postprocessing times. For simplicity, we consider a fully fault-tolerant quantum machine for quantum computing, and for classical computing, we count the number of arithmetic operations (floating point operations, FLOPS). We assume that the number of elementary logic gates associated with arithmetic is linear in the number of bits. We thus do not break the computational complexity into I/O and arithmetic operations, which is a simplification. Nonetheless, the resulting classical and quantum gate complexities are compared to provide a general sense of scaling expressed as the number of elementary operations for each architecture type. 

\vspace{0.2cm}

\textbf{Matrix inversion method.} -- Collocation equations given in Eq.~\eqref{eq:collocation_SE_simplified} can be solved by first finding the inverse of the $\mathbf B^{\dagger}\mathbf B$ matrix, followed by the Arnoldi iterative algorithm for finding eigenvalues of non-symmetric matrices. The classical computational complexity (quantified by the number of FLOPS) for finding matrix inverse is $N^\omega$, where $\omega$ is determined by the cost of matrix multiplication, with currently best value of $\omega \approx 2.371$~\cite{williams,Alman2024}. Upon recasting collocation equations into the regular eigenvalue problem written in Eq.~\eqref{eq:collocation_SE_simplified}, the Hamiltonian spectrum can be found using Krylov subspace methods (Arnoldi algorithm)
requiring $\mathcal{O}(N^2 m + Nm)$ floating point operations~\cite{Saad1986}, where $m$ is the size of the Krylov space. Typically, the precision of eigenvalues in the Arnoldi procedure increases exponentially with the size of the Krylov space, i.e., $m=\mathcal{O}(\log(1/\varepsilon))$. Thus, the total complexity for solving the collocation equations via the matrix inverse and the Arnoldi method is $\mathcal{O}\big(N^{2.371} + N^2 \log(1/\varepsilon)+N\log(1/\varepsilon)\big)$, for infinite arithmetic precision. 

However, the condition number of the $\mathbf B^{\dagger}\mathbf B$ matrix necessitates an overhead in the precision of matrix elements leading to a multiplicative increase in computational complexity of matrix inversion that grows linearly with $\log(\kappa)$, where $\kappa$ is the condition number. For the Arnoldi algorithm part, the number of iterations required to converge to eigenvalues scales as $\mathcal{O}(\kappa^{\frac{1}{2}})$~\cite{musco}. We summarize these results in Table~\ref{tab:complexity_comparison}, showing only the leading terms.
Notably, the convergence of the Krylov method does not depend directly on the size of the energy window $S$, but it does depend on the eigenvalue separation~\footnote{Note that we neglect the dependence of the Krylov space size on the number of eigenvalues requested, which is often logarithmic~\cite{Saad1986}}.

In practice, solving the non-symmetric generalized eigenvalue problem is often carried out with the generalized Schur's decomposition, with the associated cost of $\mathcal{O}(N^3)$~\cite{Golub1996}. In the Schur decomposition too, the condition number of the $\mathbf{B}$ matrix determines the minimal precision of matrix elements, thus affecting the complexity of the algorithm.

In the case of quantum-mechanical problems, the basis functions are often conveniently chosen to have non-negligible overlap (e.g., Gaussian functions). 
As the basis set size $N$ is increased, the condition number $\kappa$ of the matrix $\mathbf B^\dagger \mathbf B$ can grow uncontrollably, rendering the procedures described above infeasible. Our numerical experiments for the anharmonic one-dimensional vibrational motion Hamiltonian indicate that, with 35 Gaussian basis functions, the condition number of  $\mathbf B^\dagger \mathbf B$ is larger than $10^{16}$ (see Sec.~\ref{subsec:limitations}), making calculations with this method impractical.

Generalized eigenvalue problems in electronic-structure calculations arise from the use of non-orthogonal basis sets and, in principle, require basis orthogonalization. Such procedures, including Löwdin orthogonalization, are generally non-unitary and would therefore require additional quantum resources, making them undesirable to perform on a quantum computer unless strictly necessary. Orbital orthogonalization is typically applied to a modest number of localized functions and is therefore best treated as a classical preprocessing step. If a basis transformation must be performed within a quantum circuit, one needs to implement the non-unitary map $\ket{\xi}=S^{-1/2}\ket{\phi}$, where $S$ is the overlap matrix, for example via block-encoding of $S^{-1/2}$. Since the eigendecomposition $S=U^\dag s U$ involves only a diagonal non-unitary matrix $s$, the transformation can be implemented more efficiently than a generic basis change, for example via diagonal unitary synthesis, at a cost of $\mathcal{O}(\sqrt{N\log(1/\varepsilon)})$ for $s$.

Importantly, estimating eigenvalues of the Hamiltonian alone via QPE does not solve the Schr\"odinger equation when the Hamiltonian is represented in a non-orthogonal basis. The overlap matrix must be accounted for explicitly, either by orthogonalization or by solving the generalized eigenvalue problem directly, which is the approach adopted in this work. Non-orthogonal basis sets are nevertheless advantageous because they substantially reduce the cost of matrix-element evaluation. Gaussian basis sets in electronic-structure theory exemplify this trade-off, giving efficient analytic integral evaluation at the expense of non-orthogonality, loss of the electron-electron cusp, and the need to solve a generalized eigenvalue problem. Similar generalized eigenvalue problems naturally arise in nuclear-motion calculations and in quadrature-based discretizations of partial differential equations using distributed Gaussian functions, as considered here.

\vspace{0.2cm}

\textbf{Classical landscape scanning method.} -- In the classical landscape scanning method, we form a matrix $\mathbf M^\dagger (\alpha)\mathbf M(\alpha)$ from $\mathbf M(\alpha) = \mathbf{B}'' +\mathbf{V}^{diag}\mathbf{B} - \alpha \mathbf{B}$, given by Eq.~\eqref{eq:landscape_SE_simplified} with the cost given by $\mathcal{O}(d^2 N)$ for $d$-sparse matrices. 
For each $\alpha \in \lbrace \alpha_j\rbrace_{j=0,1,...K}$, we find the smallest eigenvalue of $\mathbf M^\dagger (\alpha)\mathbf M(\alpha)$, with the associated cost of $\mathcal{O}(dN\log(\frac{1}{\varepsilon}))$ for $d$-sparse matrices. Thus, the complete energy grid scan costs $\mathcal{O}\big(d^2N+KdN\log(\frac{1}{\varepsilon})\big)$ FLOPS respectively for sparse matrices, which transforms into $\mathcal{O}(N^3+KN^2\log(\frac{1}{\varepsilon}))$ for dense matrices. Note that $\mathbf M^\dagger (\alpha)\mathbf M(\alpha)$ is not necessarily sparse, even though $\mathbf{M}(\alpha)$ can be sparse. 
To overcome this issue, serial landscape scanning calculation can be performed with in-place memory storage for the collection of matrices $\mathbf M^\dagger (\alpha)\mathbf M(\alpha)$. 
One can simplify the calculation of the expression
\begin{equation}
\begin{split}
    &\mathbf M^\dagger (\alpha)\mathbf M(\alpha) =  \alpha^2 \,\,\mathbf{B}^\dagger\mathbf{B} - \alpha\,\big( (\mathbf{B}''+\mathbf{V}^{diag}\mathbf{B})^\dagger\mathbf{B} \\ &+ \mathbf{B}^\dagger(\mathbf{B}''+\mathbf{V}^{diag}\mathbf{B}) \big)+ (\mathbf{B}''+\mathbf{V}^{diag}\mathbf{B})^\dagger(\mathbf{B}''+\mathbf{V}^{diag}\mathbf{B})
\end{split}
\end{equation}
by precomputing three matrices that are multiplied by $\alpha^2$, $\alpha$, and the constant term. 
Therefore, at the cost of increased memory, we can speed up the calculations.
Parallelization of the energy grid scan entails a proportionally larger memory consumption. 

The quantum landscape scanning algorithm is discussed in the following subsection.

\subsection{Comparing quantum and classical algorithms}
    The complexity of quantum algorithms discussed in this work depends on four numerical factors: the size of matrices $N$, the truncation order $J$ (see Eq.~\eqref{eq:M-family}), the admissible error $\varepsilon$, and the number of grid points $K$. Condition number $\kappa$ of $\mathbf{B}$ is an artifact of the basis set choice, the number of basis functions, and grid points. 
    Classical landscape scanning algorithm finds the relevant eigenvalue for each residue matrix $\mathbf M(\alpha)$ separately~\footnote{For higher-order truncations, the complexities of the classical and quantum algorithms both grow linearly in $J$. Thus, quantum advantage can be searched already for low $J$'s.}, leading to $\mathcal{O}(JK Nd \,\text{polylog}(1/\varepsilon))$ complexity for classical landscape scanning. 

    In order to compare the quantum and classical algorithms, it is necessary to understand the relationship between parameters.
    For example, the order of the Schr\"odinger equation is linear in $\alpha$, hence $J=1$, and we shall focus on this case for the remaining part of the paper.
    Also, the number of grid points $K$ is related to the desired accuracy in $\alpha$, as the lower error we tolerate, the finer the grid is required~\footnote{In the case of vanishing derivative of the singular value dependence on the parameter $\alpha$, $K$ might be proportional to roots of $1/\varepsilon$, making this dependence better from the quantum computing perspective.}, $K \propto 1/\varepsilon$, see more detailed discussion in Sec.~\ref{subsec:algorithms_complexities}.

\textbf{Quantum landscape scanning method.} -- In the framework of digital quantum simulators, all parts of the circuit must correspond to unitary matrices.
Thus, in order to use the quantum energy landscape method summarized in Corollary~\ref{corr:quantum_algorithm_rectangular}, it is necessary to block-encode the matrix $\mathbf M = \mathbf{B}''+\mathbf{V}^{diag}\mathbf{B}$, which can be done either directly or by a linear combination of block-encodings of the component matrices: $\mathbf{B}$, $\mathbf{B}''$, and $\mathbf{V}^{diag}$. The optimal choice of block-encoding scheme depends on the particular structure of the component matrices and their block-encoding scaling constants~\cite{camps,Snderhauf2024,mukhopadhyay2024}.
In collocation, the basis sets are typically chosen such that the resulting matrix $\mathbf M(\alpha)$ is sparse. For this reason, we adopt the $d$-sparse block-encoding procedure with the associated T-gate cost $\tilde{\mathcal{O}}(N d)$~\cite{Low_2017,camps,Snderhauf2024}, where $d$ is matrix sparsity.

Following the commonly adopted consensus, we assume that the T-gates give the dominant contributions to the total quantum computing cost.
Corollary~\ref{corr:quantum_algorithm_rectangular} states that one must call the block-encoding circuit $\mathcal{O}(\sqrt{NK}/\varepsilon)$ times, giving the total T-gate complexity  $\mathcal{O}(M_\text{max}\sqrt{K}N^{3/2}d/\varepsilon)$, where $M_\text{max}$ is the maximum norm of the $\mathbf M(\alpha)$ matrix, as written in Table~\ref{tab:complexity_comparison}.

The condition number of $\mathbf B^{\dagger}\mathbf B$ is determined solely by the choice of basis set and the collocation grid. Once this discretization is fixed, its value is fixed as well and is independent of the resolution of the energy grid used in the landscape scan. In contrast, $\mathbf M(\alpha)$ denotes a one-parameter family of matrices whose condition number depends on $\alpha$. 

In particular, the condition number of 
$\mathbf M(\alpha)$ is related to the separation between its two smallest singular values, which may necessitate increased numerical precision, scaling as 
$\frac{1}{\varepsilon}$. Consequently, one may expect the condition number of the residual matrix $\mathbf M(\alpha)$ to influence the required precision in quantum landscape scanning. This precision, in turn, determines the achievable resolution of the energy grid $\lbrace \alpha_j\rbrace_{j=0,1,...K}$. 

\vspace{0.2cm}

\textbf{Improvements for Gaussian basis sets} -- 
For a specific case of Gaussian basis functions utilized to formulate the collocation equations, it is possible to employ several improvements to the quantum landscape scanning algorithm. These improvements are:
\begin{enumerate}
    \item Quantum-random oracle model (QROM) implementing matrix elements of $\mathbf B$ and $\mathbf B''$ in block-encoding can be implemented with quantum arithmetic. For instance, for the collocation matrix elements $\mathbf B_{ij}$ we have the following representation:
\begin{equation}
    \mathbf B_{ij} = C e^{-(\beta i-\gamma j)^2},
    \label{eq:B-arthmetic}
\end{equation}
for some constants $\beta$ and $\gamma$ and a normalization constant $C$. 
Similarly, also $\mathbf B''$ has a compact form. The expression given in Eq.~\eqref{eq:B-arthmetic} can be approximated with polynomials with Taylor expansion. Therefore, the cost of block-encoding all matrices scales as $\mathcal{O}(\text{polylog}(1/\delta))$, for the multiplicative error $\delta$ in the elements of the unitary matrix~\cite{RuizPerez_2017}. This error is related to the accuracy of singular values of $\mathbf M(\alpha)$ in the following way $\varepsilon = \mathcal{O}(\delta \sqrt{N})$, requiring that $\delta = \mathcal{O}(\varepsilon/\sqrt{N})$. 
Due to polylog dependence, in the notation $\widetilde{\mathcal{O}}$, it does not contribute to the overall scaling.

\mbox{}\vspace{-0.3cm} % for spacing between items

\item The leading contribution to $M_\text{max}$ is expected to come from the potential energy function and $\mathbf B''$, as due to the normalization of the wavefunctions, the max-norm of matrix $\mathbf B$ is independent of its size $N$, i.e.,  $||\mathbf{B}||_\text{max} = \mathcal{O}(1)$. 
We also note that the matrices $\mathbf B''$ and $\mathbf V^{diag}\mathbf B$ are expected to share many common non-zero elements. This is most evident when Gaussian basis functions are used to construct $\mathbf B$, and considering the fact that the diagonal matrix $\mathbf V^{diag}$ multiplies the rows of the $\mathbf B$ matrix without changing the overall number of non-zero elements. Meanwhile, $\mathbf B''$ results from double differentiation of a Gaussian function, whose values still decay exponentially away from its center.
The cost of block-encoding a diagonal matrix scales as $\mathcal{O}\big(\sqrt{N\log(\frac{1}{\delta})}\Big)$~\cite{Gosset}, with $\mathcal{O}(\sqrt{N})$ additional qubits, indicating that it is advisable to perform the matrix multiplication $\mathbf V^{diag}\mathbf B$ classically and block-encode the sum $\mathbf B'' + \mathbf V^{diag}\mathbf{B}$.
In the above, $\delta$ is the multiplicative error. 

\item For equal-width distributed Gaussian basis functions, since both $\mathbf B$ and $\mathbf B''$ are circulant matrices, it suffices to encode a single column for each matrix with QROM to reduce the T-gate cost greatly.

\item If the potential $\hat V$ is given in a functional form, its cost of block-encoding scales as $\mathcal{O}(\text{polylog} (N/\varepsilon))$ via quantum arithmetic~\cite{Gosset}, however more specific resource estimation depends on the potential parameters and accuracy.

\item Gaussian basis functions are localized in space, leading to logarithmic sparsity of the residue matrix $\mathbf M(\alpha)$, in which case $d = \mathcal{O}(\log N)$. 
\item The max-norm $M_{\text{max}}$ of $\mathbf B$ is independent of the number of basis functions as we increase the dimensionality of the matrices.
What is more, the norm of the second derivative matrix $\mathbf B''$ will not grow as we go to higher dimensions in the spatial sense of the Schr\"odinger equation. 
\end{enumerate}
With the above assumptions, the quantum landscape scanning complexity is compared in Table~\ref{tab:simplified_complexity_comparison}.

\onecolumngrid
\begin{center}
\begin{table}[H]
    \centering
\begin{tabular}{|cc|c|}
\hline
\multicolumn{2}{|c|}{classical}    & quantum \\ \cline{1-3}
  matrix inversion &    \multicolumn{2}{|c|}{landscape scanning\hspace{1cm}\phantom{1}}                \\ \cline{1-3}
\multicolumn{1}{|c|}{$N^{2.371}\log(\kappa)+N^2\kappa^{\frac{1}{2}}\log(1/\varepsilon)$} & $KN \log(1/\varepsilon)$ & $\sqrt{KN}\frac{1}{\varepsilon}\,\text{polylog}\, (1/\varepsilon)$     \!\!\!\!\!\!\!\!\! \phantom{$\frac{\int}{\int}$}            \\ \hline
\multicolumn{1}{|c|}{condition number problem} & --- & ---           \\ \hline
\end{tabular}
\caption{Classical and T-gate complexities $\widetilde{\mathcal O}$  for different methods solving the collocation Schr\"odinger equation. Given are classical computing and quantum computing methods, with classical computing methods divided into the matrix inversion technique and landscape scanning discussed in Sec.~\ref{sec:landscape_scanning}. Complexities for the classical computing algorithms are expressed in the number of floating-point operations while quantum computing algorithm has its complexity expressed in the number of non-Clifford T-gates. The table was generated accounting for Gaussian basis sets with simplifications and improvements discussed in the main text for the quantum landscape scanning algorithm. 
Using the $\widetilde{\mathcal O}$ notation, we were able to disregard unimportant $\log N$ factors in the classical landscape scanning, simplifying $\mathcal{O}\big(N\log^2N + KN\log N\log(1/\varepsilon)\big)$ to $\widetilde{\mathcal{O}}\big(KN \log(1/\varepsilon)\big)$.
Here $N$ is the number of basis functions, $K$ is the number of grid points in the landscape method $\lbrace \alpha_j\rbrace_{j=1,2,...,K}$, $\varepsilon$ is the error in the determination of the energies, and $\kappa$ is condition number of $\mathbf B^{\dagger}\mathbf B$.
}\label{tab:simplified_complexity_comparison}
\end{table}
\end{center}
\twocolumngrid

\section{Discussion}\label{sec:discussion}
There are several parameters that govern the problem of determining energies: the target accuracy $\varepsilon$, the number of basis functions $N$, and the number of grid points $K$. The scaling with respect to $\varepsilon$ is most favorable for the matrix-inversion algorithm. However, this method may not be applicable if the matrix inversion problem is ill-posed.

In contrast, landscape-based algorithms, both classical and quantum, perform better with the number of basis functions $N$, scaling as $\mathcal{O}(N^{1/2})$ and $\mathcal{O}(N)$, respectively, under the assumption of inexpensive QROM access, as discussed in the previous section. Notably, the number of grid points $K$ does not directly affect the matrix-inversion algorithm, which complicates direct comparisons among the three methods.

The quantum algorithm described above, through its extension to a larger parameter space and the use of the Quantum Phase Estimation procedure, achieves quadratically better scaling with respect to both the system's dimension and the number of grid points. However, this comes at the cost of worse scaling with respect to the precision: $\mathcal{O}\big(\log(K)\log⁡(KN)\sqrt{KN}\log⁡(1/\varepsilon)/\varepsilon\big)$, see discussion in App.~\ref{app:accuracy}.

For our algorithm, which uses a relatively inexpensive state preparation of an equal superposition of $NK$ singular vectors, the resulting complexity matches that of Grover search over the singular values $\widetilde{\mathbf{M}}$. Consequently, the  $\mathcal{O}(\sqrt{NK})$ complexity scaling in our method is believed to be optimal. However, this result does not necessarily imply optimality for finding the energies of a Hamiltonian. Any improved algorithm would require a fundamentally different state preparation unitary than the one used here and avoid amplitude amplification in the same sense.

Our method can also be applied to estimate the density of states, which is particularly relevant in solid-state physics. Information about the density of states has multiple applications. By scanning the eigenvalue grid $\lambda$ of the matrix $\widetilde{\mathbf M}$, we can construct a histogram that approximates the density of states, even for gapless systems.

The aim of this contribution is to compare, for a given choice of parameters $K$, $N$, and $\varepsilon$ that define the problem at hand, which algorithm is more efficient in terms of query complexity. It is important to note that we make no assumptions about the cost of a single query in the classical or quantum case; in particular, the latter is and will remain more expensive. While continued progress in quantum technologies is unlikely to make quantum computing cheaper per query than classical computing, we observe a faster decline in costs and a more rapid increase in capabilities for quantum technologies compared to their classical counterparts.

Therefore, this comparison should not be seen as an oracle for practical quantum advantage. Rather, it highlights the regimes and problem classes where quantum technologies have the greatest potential to outperform classical methods. What follows is a discussion of several comparative aspects of the algorithms introduced earlier.
  
\vspace{0.2cm}

\textbf{Condition number problem.} --     
    The matrix inversion method in the dense Gaussian basis (highly non-orthonormal) is unfeasible due to the required inversion of matrix $\mathbf B^{\dagger}\mathbf B$, with a high condition number.
    The same problem is less pronounced in landscape scanning, which has a better dependence on $N$ ($N$ vs.\ $N^{2.371}$) and worse dependence on $\varepsilon$ than the matrix inverse method. 
    In comparison, the quantum version of landscape scanning has an even better, quadratically improved dependence on the number of basis functions in T-gate count, at the cost of worse dependence on the accuracy $\varepsilon$ and a possible problem with the max-norm $M_\text{max}$.
    The latter problem might arise with some potentials, but there are potentials for which this will be constant with the basis function size.
    Problems with the inversion of a matrix $\mathbf A$ for a high condition number $\kappa$ occur because it is sensitive to small perturbations $\mathbf A'$, where $||\mathbf A-\mathbf A'|| \leq \epsilon$. 
    Therefore, the error of the result can be of the order of $\kappa^2 \epsilon$~\cite{derivative}.

    In summary, matrix inversion with condition number $2^k$ requires at least $2k+b$ bits of precision, where $b$ is the target precision.  An example discussed in Sec.~\ref{sec:landscape_scanning} showed that $\kappa^2$  reaches $10^{37}$ even for moderate-size Gaussian basis sets. When higher-dimensional problems are considered, this scaling is somewhat alleviated due to the distribution over a high-dimensional volume. Thus, when the condition number of $\mathbf B^\dagger \mathbf B$  is too high, then it is numerically impossible to calculate $(\mathbf B^\dagger \mathbf B)^{-1}$ with any reasonable computer precision. 
    In contrast, in the presented quantum landscape scanning algorithm, the condition number of $(\mathbf B^\dagger \mathbf B)$ is irrelevant. If $E$ is the solution to $(\mathbf H-E \mathbf B)\ket\psi=0 $ and $\|\mathbf H'-\mathbf H\|$, $\|\mathbf B'-\mathbf B\|$,  $|E'-E|\leq\varepsilon$, then minimal singular value of $\mathbf H'-E' \mathbf B' $ is $\Od{\varepsilon}$.

    \vspace{0.2cm}

    \textbf{The number of grid points $K$.} -- 
    In the matrix-inversion algorithm, there is no need to divide the energy window into a grid; thus, there is no direct dependence on this parameter. 
    On the other hand, both classical and quantum landscape scanning algorithms require discretization of the parameter space. 

    This division depends substantially on the width of the dips in the energy landscape, as shown in Figs.~\ref{fig:26functions}-\ref{fig:35-36functions}. Because we aim to measure at least one grid point $\alpha$ capturing the corresponding minimum, the choice of grid must be adjusted accordingly. The widths of the minima depend on the number of basis functions. To successfully find the energies using both landscape methods, it might be necessary to adaptively choose precision $\varepsilon$ of singular values, i.e., the number of grid points $K$ for a fixed energy window $S$.

\vspace{0.2cm}
    
\textbf{Potentials with a potential quantum advantage.} -- 
The quadratic improvement in computational complexity with respect to the matrix size $N$ offered by our quantum algorithm, compared to classical methods, must be considered alongside its scaling with the admissible error level $\varepsilon$ (i.e.,  eigenvalue precision). Even with the advent of sufficiently powerful quantum computers, quantum algorithms will not offer an advantage in all cases.

While the quantum algorithm exhibits more favorable scaling in $N$, typically $\mathcal{O}(\sqrt{N})$ T-gates, compared to $\mathcal{O}(N)$ FLOPs for classical methods (see Tables \ref{tab:complexity_comparison} and ~\ref{tab:simplified_complexity_comparison}) the classical algorithm often has better scaling in precision, e.g., $\mathcal{O}(\log(1/\varepsilon))$. This implies that quantum advantage is most likely to be realized in scenarios where high-dimensional quantum systems are modeled, many eigenvalues are required (large $N$), but only moderate precision is acceptable.

Thus, relative to classical landscape scanning, the quantum algorithm is particularly beneficial when the number of grid points $K$ is large and the error tolerance $\varepsilon$ is not extremely small. This is common in systems with shallow potential energy wells and dense spectra, such as potential energy surfaces of polyatomic molecules. In such systems, the density of eigenvalues increases rapidly with energy, making it difficult for classical methods to recover all highly excited states. Our quantum algorithm, by contrast, enables estimation of the density of states with controllable resolution and favorable scaling in the number of histogram bins $K$. These advantages also suggest potential applicability to quasi-continuum problems in solid-state physics, where classical methods similarly encounter limitations.

\vspace{0.2cm}

\textbf{Sampling problem.} -- 
To construct an energy histogram using the quantum or classical landscape scanning algorithm, multiple runs are required, as each execution yields only a single eigenvalue. Quantum landscape scanning however, is defined over an extended Hilbert space encoding candidate energies $\alpha$ in quantum states. Utilizing quantum superposition and amplitude amplification, fewer calls to eigenvalue computation are required than in the classical computing case. Specifically, the probability of detecting a given “dip” in eigenvalues of the residue matrix is proportional to its width. In systems like the harmonic oscillator or Morse potential, dip widths can vary by an order of magnitude across energy scales. This implies that to avoid oversampling higher energies, it may be useful to partition the energy window $S$. On the other hand, if high-energy states are of particular interest, this behavior becomes advantageous, as low-energy states are typically easier to access by other means.

Assuming the widths are roughly uniform, the total number of runs required to resolve all $L$ dips scales as $\mathcal{O}(L \log L)$. This introduces an overhead in both quantum and classical landscape scanning algorithms, whereas in the matrix inverse method, this overhead is concealed under the shift to the Hamiltonian and the size of the Krylov basis. Nevertheless, for realistic scenarios with $L < 100$, this overhead remains modest.

\vspace{0.2cm}

\textbf{QSVT-based bandpass filter as an alternative.} -- 
An alternative approach based on quantum singular value transformation (QSVT) can be used to construct band-pass filters acting on each block $\mathbf M(\alpha)$, localized near a target eigenvalue $\lambda_0$ \cite{Lin2020,Martyn2021}. For the extraction of a single parameter value, such a method can achieve a complexity comparable to QPE combined with amplitude amplification. However, in landscape-scanning problems involving many values of $\alpha$, QSVT-based filters act independently on each block and require either repetition over grid points or additional amplification overhead. By contrast, the QPE+AA construction employed here enables global selection in superposition, directly amplifying the values of $\alpha$ associated with eigenvalues in a prescribed spectral window of $M(\alpha)$.

\section{Concluding remarks}\label{sec:conclusions}
In this work, we presented a quantum algorithm for computing eigenvalues and singular values of families of parametrized matrices. We showed that our method achieves an up to quadratic speedup over both classical algorithms and the direct application of Quantum Phase Estimation to the matrix family.

In particular, we demonstrated the usefulness of our approach for solving the Schrödinger equation, with the most significant advantages observed in scenarios where: (a) multiple eigenvalues are required; and (b) the matrix representation of the Schrödinger equation takes the form of a generalized eigenvalue problem.

The core idea of our technique is to convert the generalized eigenvalue problem into a task of scanning (quantum landscape scanning) the singular values of an appropriately defined residue matrix as a function of the matrix family parameter, and then identifying the minima. We adapt this procedure to fault-tolerant digital quantum simulators and illustrate its application using the pseudospectral collocation method for solving differential equations, including the Schrödinger equation.
Unlike classical collocation, our method does not rely on matrix inversion, which enables the computation of highly excited states in molecular systems, where standard methods often struggle due to large condition numbers. We demonstrated this capability using two analytically solvable models: the harmonic oscillator and the Morse potential.

Our quantum landscape scanning method scales favorably with the number of basis functions $N$: it achieves $\mathcal{O}(N^{1/2})$ T-gate complexity, compared to at least $\mathcal{O}(N)$ FLOPS scaling for both classical landscape scanning and matrix inversion approaches. This speedup arises from encoding candidate energy levels (i.e.,  parameters of the residue matrix family) into a quantum superposition and applying amplitude amplification simultaneously in the space of candidate eigenvalues and the spectrum of the residue matrix. While the classical matrix inversion method offers better scaling in eigenvalue precision $\varepsilon$, it becomes impractical in the high-energy regime due to rapidly increasing condition numbers.

Our conclusion is that detailed investigations of specific physical problems, especially those well-represented by collocation methods, should be pursued to quantitatively assess the advantages of quantum computing. While our results indicate favorable polynomial quantum speedups, care must be taken to account for context-dependent evaluation versus classical computing techniques.

\mbox{}\\

\textbf{Acknowledgments.} --
We acknowledge funding from the European Innovation Council accelerator grant COMFTQUA, no. 190183782. The authors are grateful to Tucker Carrington for the initial discussions on the classical landscape method.

\bibliography{bibliography}

\appendix

\section{Remarks on the precision}\label{app:accuracy}
Thm.~\ref{thm:quantum_algorithm_square} consists of two main phases: internal QPE and subsequent amplitude amplification. 
For AA, determination of the optimal number of necessary steps $\xi$ is related to the number of grid points that satisfy our condition. 
Nonetheless, since this is a priori not known, it is possible to fall back to a naive method: choose a random number $\xi_1$ from $\Big\{1, \cdots, \left \lfloor{\pi/4\sqrt{NK}}\right \rfloor\Big\}$. 
The amplitude amplification after $\xi_1$ Grover rounds has expectation value $\mathcal{O}(1)$. 
Therefore, after applying this subroutine a constant number of times for different random $\xi_i$, we can boost the amplitude to any number close to 1, achieving the quadratic speedup.
Essentially, this corresponds to a standard Grover search with an unknown number of marked states.

For QPE, there are two hidden parameters that we have not discussed before, $\omega$ and $\iota$. 
The first, $\omega<1/2 $ is defined such that the probability of obtaining a correct answer, after application of our algorithm, is at least $1-\omega$. 
The second number, $\iota $, determines the precision: if we measure for the state $\ket{\Gamma^{(0)}} $ the qubit $\chi$ with result $\ket{1}_\chi$ then, as the result of the QPE we obtain a number $\lambda_i$.
This number approximates the true eigenvalue, up to accuracy $\iota \varepsilon$.
Similarly, if the measured state is  $\ket{0}_\chi$, then the corresponding eigenvalue is outside of the said range $\pm \iota \varepsilon$.
The number $\iota $ can be adjusted to the total accuracy, yielding costs to be
\begin{equation}\label{eq:without_omega}
    \mathcal{O}\bigg( \frac{\sqrt{NK}}{\iota\varepsilon}\log(KN)\log\big(1/(\iota\varepsilon)\big)\bigg),
\end{equation}
under the assumption of constant $\omega$.

Moreover, to apply QPE, we need to rescale the spectrum.
We do so by transforming the original operator $\widetilde{\mathbf M}$ to $e^{i(c\widetilde{\mathbf M}+d\,\mathbb{I})}$ where $c,d\in\mathbb{R}$ are chosen in such way that the eigenvalues of $c\,\widetilde{\mathbf M}+d\,\mathbb{I}$ are in $[\iota\varepsilon,2\pi-\iota\varepsilon]$.
Here, we choose $c$ as large as possible to have the best resolution.
In turn, we can choose $\iota =1/4$ to find small eigenvalues properly. 

Expression~\eqref{eq:without_omega} assumes constant $\omega$, but it might be an important factor, depending on the usage of the algorithm.
Using naive QPE, we could obtain the probability that the phase is wrong to be smaller than $\omega/(KN)$ as there are $KN$ eigenvectors in total. 
Then, the cost of obtaining the state before amplitude amplification $\ket{\Gamma^{(0)}}$ would be $\Od{\log{(KN/\omega)}\log(1/\varepsilon)/\varepsilon}$. 
However, to obtain $\alpha_i$ with a probability higher than 1/2, a better way would be to apply an iterative approach.
First, we could choose parameters so that $\omega$ is not small -- then one could check if $\widetilde{\mathbf M}$ has an eigenvalue in the interval of interest. 
The number of correct eigenvalues is asymptotically proportional to the number of grid points $K$.   
Thus, the check costs  $\Od{\log{(N)}\sqrt{N}\log(1/\varepsilon)/\varepsilon}$, without the $K$-dependence.
To reach probability $1-\omega$ for an arbitrary $\omega$, we would repeat it $\Od{\log(1/\omega)}$. 
Therefore, the entire cost is
\begin{equation}
    \mathcal{O}\Big((\sqrt{K}\log(KN) +\log(1/\omega)\log N)\sqrt{N}\log(1/\varepsilon)/\varepsilon\Big),
\end{equation} 
making the connection between $K$ (stemming from AA) and $\omega$ (QPE part) additive, as opposed to multiplicative, coming from the naive approach.
We omit this technique in the main text as the scaling is equivalent from the asymptotic $\widetilde{\mathcal{O}}$ perspective.

With the median trick~\cite{Kerzner_2024}, we are able to lower $\omega$ at acceptable costs of increasing the number of controlled qubits in QPE $b$. 
By applying QPE to $c$ different registers, we multiply the total qubit count to $cb$, but lowering the probability of error. 
If, for a single QPE call, the $\omega_{QPE}$ is constant and smaller than 0.5, then using median trick $c$ times reduces it in an exponential fashion, so $c = \mathcal{O}\big(\log(1/\omega_{QPE})\big)$.

\section{Background removal}\label{app:background_removal}
In the landscape scanning method for solving the collocation problem, an improvement of the sensitivity of the cutoff accuracy $\varepsilon$ can be found via ``background'' removal. 
By this, we mean recognizing that the minimal singular value of a linear combination of two matrices depends approximately linearly on their mixing coefficient, $\min \sigma(\mathbf Z + x \mathbf Y) \approx ax+b$, for $x$ in our region of interest.

To exploit this, we first determine the dependence of the singular value on the parameter $\alpha$ (i.e., the slope $s$ of the curve in $\alpha = 0$) and remove it, thereby enhancing the method’s sensitivity to localizing minima in the singular value. 
\begin{equation}\label{eq:background_removal}
    \sigma_\mathrm{min}(\alpha) \mapsto \sigma'_\mathrm{min}(\alpha) \coloneqq  \sigma_\mathrm{min}(\alpha) - s \alpha.
\end{equation}

Overall, the action of this transformation is presented in Fig.~\ref{fig:background_removal}.

\onecolumngrid
\begin{center}
\begin{figure}[h!]
\begin{tikzpicture}
        \node at (0,0) {\includegraphics[width=0.99\linewidth]{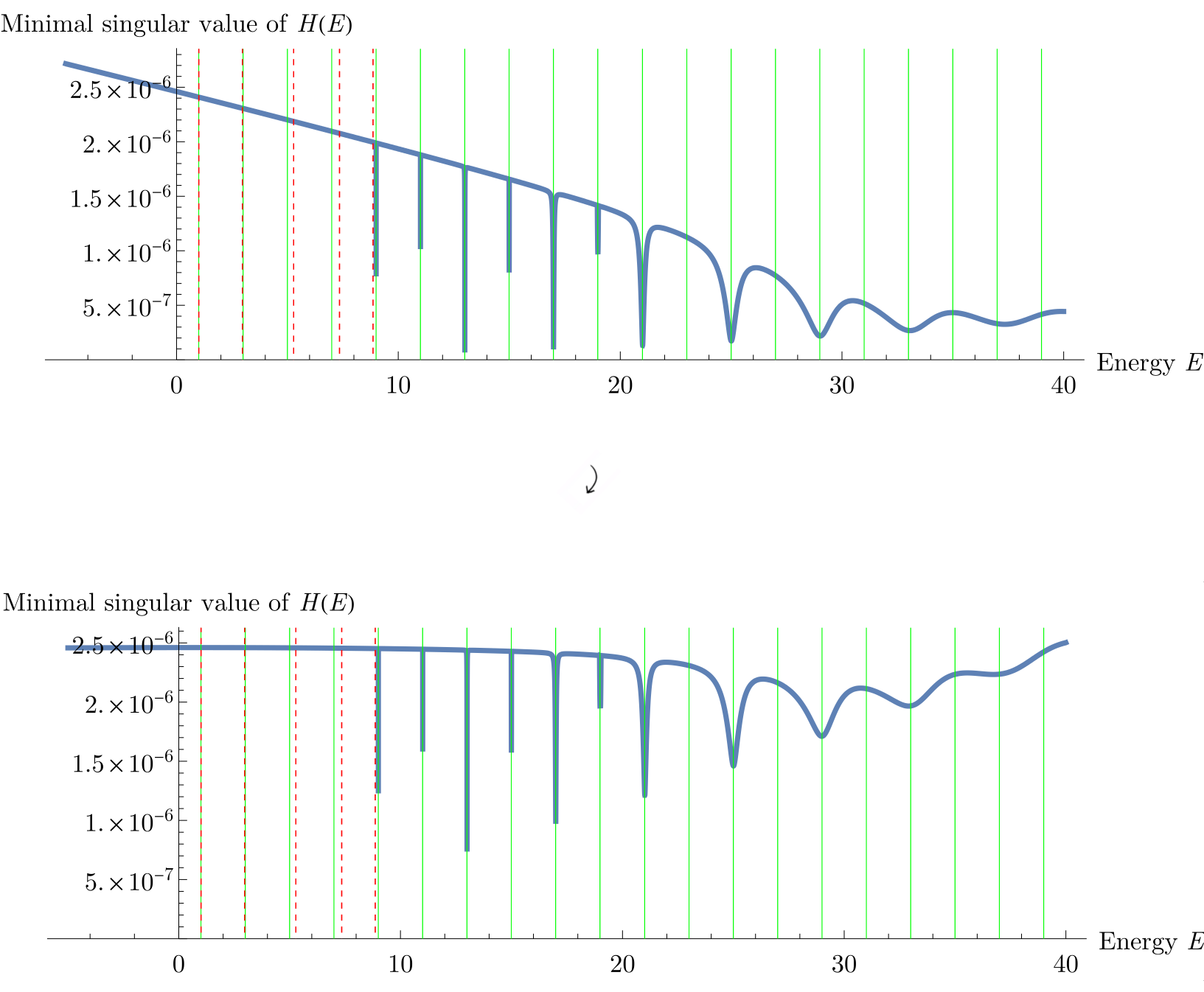}};
        
        \node at (-9.1,7.5) {(a)};
        \node at (-9.1,-1.) {(b)};
\end{tikzpicture}
        \caption{We extend the validity of the cutoff $\varepsilon$ via background removal.
    In (a), the minimal singular value is found straightforwardly, while in (b), we perform the transformation given by Eq.~\eqref{eq:background_removal}.
    Observe that, without this transformation, different values of $\varepsilon$ must be chosen for energies below 12 and above 20, while the background removal allows for a global choice.
    In this figure, we study a harmonic oscillator potential with 26 basis functions and 80 grid points.}\label{fig:background_removal}
\end{figure}
\end{center}

\end{document}